\documentclass[a4paper,fleqn]{cas-sc}

\usepackage[authoryear]{natbib}

\def\tsc#1{\csdef{#1}{\textsc{\lowercase{#1}}\xspace}}
\tsc{WGM}
\tsc{QE}

\usepackage{float}
\usepackage{geometry}
\usepackage{array}
\usepackage{booktabs}
\usepackage{graphicx}
\usepackage{bm}
\usepackage{xcolor}
\usepackage{caption}
\usepackage{amsthm}
\usepackage{algorithm}
\usepackage{algpseudocode}
\usepackage{caption}
\usepackage{subcaption}
\usepackage{natbib}
\usepackage{accents}
\usepackage{placeins} 
\usepackage{longtable}
\usepackage{float}

\newtheorem{theorem}{Theorem}
\newtheorem{lemma}{Lemma}

\newtheorem{remark}{Remark}
\usepackage{hyperref}
 \hypersetup{
    colorlinks=true,
    linkcolor=blue,
    filecolor=blue,      
    urlcolor=blue,
    citecolor=blue
}
\usepackage{enumitem}
\numberwithin{equation}{section}
\newcounter{assump} 
\renewcommand{\theassump}{AN\arabic{assump}}

\def\tsc#1{\csdef{#1}{\textsc{\lowercase{#1}}\xspace}}
\tsc{WGM}
\tsc{QE}
\tsc{EP}
\tsc{PMS}
\tsc{BEC}
\tsc{DE}

\usepackage{amssymb}

\begin{document}
\let\WriteBookmarks\relax
\def\floatpagepagefraction{1}
\def\textpagefraction{.001}
\shorttitle{}
\shortauthors{Sarkar et~al.}
\title[mode=title]{\textbf{A Three-Stage PCA Procedure for Sequentially Arriving High-Dimensional Data}}  



%
\author[1]{Partha Sarkar}[orcid=0009-0006-7526-1362] 
\cormark[1]
\author[2]{Sairam Rayapolu}[orcid=]
\author[3]{Bhargab Chattopadhyay}[orcid=0000-0001-7713-3125]
\affiliation[1]{organization={Department of Statistics}, Department and Organization
            addressline={Florida State University,}, 
           city={Tallahassee},
           state={Florida},
            postcode={32301},
            country={USA}}
            \affiliation[2]{organization={Data Science and Information Systemss},Department and Organization
            addressline={Krea University,}, 
           city={Sri City},
           state={Andhra Pradesh},
            postcode={517646},
            country={India}}
\affiliation[3]{organization={School of Management and Entrepreneurship},Department and Organization
            addressline={Indian Institute of Technology Jodhpur,}, 
           city={Karwar},
            postcode={342030},
            country={India}}
\cortext[cor1]{Corresponding author: ps24v@fsu.edu}

















\begin{abstract}
We develop a three-stage adaptive procedure for principal component analysis (PCA) when high-dimensional observations are collected sequentially and additional sampling incurs a cost. The procedure balances PCA compression loss against sampling cost while selecting the retained dimension through a prescribed explained-variance criterion. Starting from a pilot sample, an intermediate stage updates the PCA quantities before determining the final sample size, thereby avoiding reliance on unknown population eigenvalues. Under suitable regularity conditions, we establish both first- and second-order efficiency relative to the population oracle. Comparison with the corresponding two-stage rule shows that the additional recalibration yields sharper second-order control and reduces the influence of the pilot stage on the final sampling decision. The theory allows the ambient dimension to exceed the sample size under appropriate covariance and spectral conditions. Simulation studies demonstrate the strong finite-sample performance of the procedure across increasing dimensions and several dense covariance structures. As a real-data application, we conduct a retrospective study of gene-expression data from 32 cancer-type cohorts in The Cancer Genome Atlas, illustrating both cost-effective early stopping and settings in which additional observations are recommended.
\end{abstract}


\maketitle


\section{Introduction}
\label{sec_intro}

\noindent Principal component analysis (PCA) is one of the most widely used methods for reducing the dimension of multivariate data by representing observations through a smaller set of directions that retain most of the variation \citep[e.g.,]{jolliffe2016principal}. Its usefulness is particularly evident in modern high-dimensional applications, such as gene-expression studies, where each observation may contain measurements on tens of thousands of variables (\(p\)) while the number of available subjects (\(n\)) is comparatively small \citep[e.g.,]{johnstone2009,aoshima2018survey}. In many such studies, observations are also collected sequentially and at a non-negligible cost. Additional sampling can improve estimation of the covariance structure and the resulting principal components, but it also increases study cost. This leads to a natural sequential question: how long should sampling continue while updating the PCA representation needed to achieve a prescribed level of explained variation?

The sampling decision is closely tied to the unknown covariance structure. The spectrum determines both how many components are needed to attain the target level of explained variation and how much residual variation remains after compression, which in turn governs the benefit of collecting more observations. This dependence on the unknown spectrum becomes more challenging when \(p\) grows with \(n\) and may greatly exceed it. In such regimes, PCA behavior depends strongly on the underlying spectral structure and on how population variation is distributed across the principal directions \citep[e.g.,][]{paul2007asymptotics,aoshima2018survey,ReissWahl2020}. The purpose of this paper is therefore to develop a sequential procedure that determines the required sample size while updating the retained PCA dimension as new observations are acquired.

Several strands of literature are related to this problem. One concerns PCA for streaming or sequentially arriving data. The recursive learning rule of \citet{Oja1982} is an early foundation for online PCA, and subsequent work has developed finite-sample guarantees and computationally efficient methods for updating principal components as new observations arrive \citep[e.g.,][]{JainEtAl2016,cardot2018online}. Online formulations based on cumulative compression or reconstruction loss have also been studied; see, for example, \citet{nieonlinePCA}. These methods primarily address how the principal components should be updated along a data stream. Our problem contains an additional decision: when observations themselves are costly, one must also determine when data collection should stop.

\noindent A second line of work concerns the behavior of PCA in high dimensions. Classical PCA intuition can change substantially when the dimension and sample size grow together, as demonstrated in the spiked covariance setting by \citet{paul2007asymptotics}; related difficulties with standard PCA when the dimension is comparable to or much larger than the sample size are discussed by \citet{johnstone2009} and surveyed by \citet{aoshima2018survey}. More generally, the accuracy of sample covariance estimation can depend on the effective dimension of the covariance structure rather than on the ambient dimension alone \citep[e.g.,][]{koltchinskii2017}, while PCA reconstruction accuracy is similarly governed by the underlying eigenvalue structure \citep[e.g.,][]{ReissWahl2020}. This perspective is particularly relevant to the application in Section~\ref{application}, where the number of measured variables is much larger than the sample size, yet the dominant variation can still be concentrated in a relatively small number of principal directions.

The sampling aspect of our problem is closely related to the classical literature on sequential and multistage estimation, where an initial sample is used to obtain a preliminary estimate of the required sample size and additional stages are introduced to refine that estimate as more information becomes available \citep[e.g.,][]{hall1981asymptotic,ghosh201997sequential, chattopadhyay2025sequential}. Within PCA, cost-compression criteria have been studied through procedures that balance the loss incurred by dimension reduction against the cost of collecting additional observations \citep[e.g.,][]{chaban2022,BanerjeeChattopadhyay2025}. These earlier developments primarily consider two-stage procedures in fixed-dimensional settings; the available second-order analysis also assumes a prespecified retained dimension and uses additional population spectral information in calibrating the sampling rule. Our setting is different in that the ambient dimension is allowed to increase with the sample size and the number of principal components is itself selected from the accumulating data according to an explained-variance criterion. Thus, uncertainty in the covariance structure affects both the required sample size and the dimension of the final PCA representation.

To address these features, we introduce the \textit{three-stage procedure} in Algorithm~\ref{algo_1}. The pilot sample is used only to determine an intermediate sample size rather than the final target directly. After the additional observations are collected, the covariance structure, retained PCA dimension, and residual variation are re-estimated using the enlarged sample, and this updated information is then used to determine the final sample size. The additional recalibration therefore reduces the influence of estimation error at the pilot stage and allows both the dimension-selection and sample-size decisions to adapt to the information accumulated during the study. Importantly, the procedure is fully data driven and does not require population eigenvalues or other unknown spectral quantities as inputs.

Our theoretical analysis shows that, under the standard regularity conditions, the proposed procedure achieves both \textit{first- and second-order efficiency} relative to the population oracle, even in \textit{high-dimensional settings}. Theorems~\ref{th:th_almost_sure_cons}--\ref{th:risk_cons} establish first-order efficiency through almost sure sample-size consistency, mean sample-size efficiency, and risk efficiency, while Theorems~\ref{th:mean_diff} and~\ref{th:risk_diff} establish second-order efficiency by showing that the expected stopping sample size remains within a bounded distance of the oracle target and that the risk regret is of the same order as the sampling cost. A key advantage of the three-stage construction is the final recalibration: for the corresponding two-stage rule, both the expected sample-size discrepancy and the risk regret retain an additional factor involving the ratio of the oracle target to the pilot size and can therefore be substantially larger when the pilot is small relative to the oracle target; see Section~\ref{two_stage_comparison}. The third stage removes this pilot-dependent inflation, yielding the sharper second-order efficiency of the proposed procedure.

We examine the empirical performance of the proposed method through simulation studies with increasing ambient dimensions and several dense covariance structures exhibiting different patterns of eigenvalue decay. The simulations assess sample-size selection, risk performance, and recovery of the population PCA dimension, providing a strong finite-sample illustration of the theoretical results.

As our real-data application, we conduct a retrospective study using gene-expression data from 32 cancer-type cohorts in The Cancer Genome Atlas \citep[e.g.,][]{Weinstein2013,Hoadley2018,GDCDataPortal2026}. Each patient has measurements on approximately \(20{,}000\) genes, while the number of available patients is far smaller, providing a natural high-dimensional setting. We first examine whether most of the observed variation is concentrated in a relatively small number of principal directions and then retrospectively apply our \textit{three-stage procedure}, under a specified sampling-cost trade-off, to determine a data-driven target sample size for each cohort. The results illustrate both possible outcomes: for larger cohorts, the procedure can recommend stopping well before all available patients are used, whereas for smaller cohorts the recommended target can exceed the available sample size, indicating that additional observations would be desirable.

\noindent The remainder of the paper is organized as follows. Section~\ref{sec:notation} introduces the notation and basic definitions, and Section~\ref{sec:setup_oracle} formulates the statistical framework and oracle benchmark. Section~\ref{sec:three_stage} presents the proposed three-stage procedure, Section~\ref{sec:asym_efficiency} develops its first- and second-order efficiency theory and comparison with two-stage sampling, Section~\ref{sec_simu} reports the simulation studies, and Section~\ref{application} presents the TCGA real-data application. Proofs of the main results and additional technical results are provided in Appendices~\ref{app_A} and~\ref{app_B}.

\subsection{Notation and definitions}\label{sec:notation}

\noindent Let us introduce the notations and definitions used throughout the paper. For positive sequences $a_n$ and $b_n$, we write $a_n = O(b_n)$ if there exists a constant $C$ such that $a_n \leq C b_n$ for all $n \in \mathbb{N}$, and $a_n = \Omega(b_n)$ if there exists a constant $C$ such that $a_n \geq C b_n$ for all $n \in \mathbb{N}$. We write $a_n = o(b_n)$ if $\lim_{n \to \infty} a_n/b_n = 0$, and $a_n \sim b_n$ if $a_n/b_n \to 1$ as $n \to \infty$.

For matrices and vectors, $I_p$ denotes the identity matrix of order $p$, and for a symmetric $p \times p$ matrix $A$, $\lambda_1(A) \geq \lambda_2(A) \geq \cdots \geq \lambda_p(A)$ denote its ordered eigenvalues. The \textit{effective rank} of a positive semi-definite matrix $\Sigma_0$ \citep[e.g.,][]{koltchinskii2017} is defined by $\widetilde{r}(\Sigma_0) = \operatorname{tr}(\Sigma_0)/\lambda_1(\Sigma_0)$. We write $\mathcal{S}^{p-1}$ for the unit Euclidean sphere in $\mathbb{R}^p$. For a vector $x \in \mathbb{R}^p$, its $r$-th norm is $\|x\|_r = \big(\sum_{j=1}^p |x_j|^r\big)^{1/r}$, with $\|x\|$ denoting the Euclidean norm. For a $p \times p$ matrix $A = (A_{ij})_{1 \leq i,j \leq p}$, the spectral norm is $\|A\| := \sup_{u \in \mathcal{S}^{p-1}} \|Au\|$. For a scalar random variable $Y$, we denote its sub-Gaussian norm by $\|Y\|_{\psi_2}$, and for a random vector $X\in\mathbb{R}^p$ we use $\|X\|_{\psi_2}:=\sup_{u\in\mathcal{S}^{p-1}}\|\langle X,u\rangle\|_{\psi_2}$ \citep[e.g.,][]{Vershynin2018,vershynin2010introduction}.

\section{Statistical Framework and Oracle Benchmark}\label{sec:setup_oracle}

\noindent In this section, we introduce the statistical framework and the PCA targets used throughout the paper, formulate the cost-compression criterion \cite{chaban2022} in the high-dimensional setting, and define the corresponding population oracle that serves as the benchmark for the sequential procedure. We then state the regularity conditions and establish the basic oracle properties needed for the development and analysis of the proposed three-stage methodology.

\subsection{Statistical setup and PCA targets}\label{setup_target}

\noindent Suppose that \(X_1,\ldots,X_n\) are \(p_n\)-dimensional independent and identically distributed mean-zero sub-Gaussian random vectors \citep[e.g.,][]{Vershynin2018}, with
\(\operatorname{Var}(X_1)=\Sigma_{0n}\).
The notation \(p_n\) emphasizes that the ambient dimension may grow with the sample size \(n\). Accordingly, the population covariance matrices
\(\Sigma_{0n}\in\mathbb R^{p_n\times p_n}\) are also allowed to vary with \(n\). For a prescribed explained-variance proportion \(\eta\in(0,1)\), define
\begin{equation}
k_{0n}
:=
\min\left\{
k\in\{1,\ldots,p_n\}:
\sum_{j=1}^{k}\lambda_j(\Sigma_{0n})
\ge
\eta\,\operatorname{tr}(\Sigma_{0n})
\right\}.
\label{eq:k0n}
\end{equation}
Thus, \(k_{0n}\) is the smallest number of population principal components required to explain at least \(100\eta\%\) of the total population variance.

\noindent The natural empirical counterpart of \(\Sigma_{0n}\) is the sample covariance matrix \(S_n=n^{-1}\sum_{i=1}^{n}X_iX_i^\top\). The corresponding empirical counterpart of \(k_{0n}\) is denoted by \(\hat{k}_n\). Specifically, \(\hat{k}_n\) is the smallest number of sample principal components required to explain at least \(100\eta\%\) of the total sample variance, that is,
\begin{equation}
\hat{k}_n
:=
\min\left\{
k\in\{1,\ldots,p_n\}:
\sum_{j=1}^{k}\lambda_j(S_n)
\ge
\eta\,\operatorname{tr}(S_n)
\right\}.
\label{eq:khatn}
\end{equation}

\noindent Let \(u_{1n},\ldots,u_{p_n n}\) and
\(\hat{u}_{1n},\ldots,\hat{u}_{p_n n}\) be orthonormal eigenvectors of \(\Sigma_{0n}\) and \(S_n\), respectively, ordered according to their corresponding eigenvalues. For \(k\in\{1,\ldots,p_n\}\), define the population and empirical projection matrices by
\[
P_{0n}^{k}
:=
\sum_{j=1}^{k}u_{jn}u_{jn}^{\top},
\qquad
P_n^{k}
:=
\sum_{j=1}^{k}\hat{u}_{jn}\hat{u}_{jn}^{\top}.
\]
These matrices project onto the subspaces spanned by the first \(k\) population and sample principal component directions, respectively. The variance-based selection rules therefore yield the projections
\(P_{0n}^{k_{0n}}\) and \(P_n^{\hat{k}_n}\). Theerefore, the population residual variance after retaining \(k_{0n}\) components and its empirical counterpart are defined, respectively,
\begin{equation}\label{V_Xi}
\xi_n:=\operatorname{tr}\!\left[\left(I_{p_n}-P_{0n}^{k_{0n}}\right)\Sigma_{0n}\right]
=\sum_{j=k_{0n}+1}^{p_n}\lambda_j(\Sigma_{0n}),
\qquad
V_n(\hat{k}_n):=\operatorname{tr}\!\left[\left(I_{p_n}-P_n^{\hat{k}_n}\right)S_n\right]
=\sum_{j=\hat{k}_n+1}^{p_n}\lambda_j(S_n).
\end{equation}
Thus, \(\xi_n\) and \(V_n(\hat{k}_n)\) quantify the population and sample variance, respectively, left unexplained by the retained principal components, with empty sums interpreted as zero. We also define the population and sample explained-variance curves by \(\gamma_k:=\sum_{j=1}^k\lambda_j(\Sigma_{0n})/\operatorname{tr}(\Sigma_{0n})\) and \(\hat\gamma_k:=\sum_{j=1}^k\lambda_j(S_n)/\operatorname{tr}(S_n)\), respectively. 

\subsection{Cost-compression loss and risk}\label{cost_compression}
\noindent Unlike fixed sample size procedures, sequential procedures accumulate the data in stages and the final sample size is determined using a stopping rule. This is especially useful when each additional observation carries a real cost, as in applications involving expensive data acquisition. In such settings, the sample size can then be chosen to balance the sampling cost against the value of collecting more data. For example, in Section~\ref{application} we analyze gene-expression data from The Cancer Genome Atlas (TCGA) program. Genetic and genomic testing in oncology can be costly, with reported costs ranging from \$76.91 to \$11,431.66, and whole-exome sequencing (WES) and whole-genome sequencing (WGS) among the most expensive options \citep[e.g.,][]{Goh2026}. Thus, collecting more observations may improve estimation of the covariance structure and stabilize the estimated number of principal components, but it also increases the total sampling cost. To capture this trade-off, we use the cost-compression loss introduced by \cite{chaban2022},
\begin{equation}
L_n(\hat{k}_n)
=
\frac{A\,V_n(\hat{k}_n)}{n}+cn,
\qquad A>0,\quad c>0,
\label{eq:cost_compression_loss}
\end{equation}
where \(V_n(\hat{k}_n)\) is the sample residual variance defined in \eqref{V_Xi}. This is a natural loss function for our setting, since the goal is to choose the sample size while performing PCA-based dimension selection. The quantity \(V_n(\hat{k}_n)\) measures the amount of variation not captured by the retained principal components. With a larger sample, this residual variation is assessed using more information, and its contribution to the overall statistical loss is therefore discounted by the factor \(1/n\), with its relative importance controlled by \(A>0\). Thus, increasing \(n\) reduces the effective penalty associated with PCA compression, but only at the cost of collecting additional observations. The term \(cn\), with \(c>0\) representing the per-observation cost, prevents the procedure from increasing the sample size indefinitely. The resulting criterion therefore seeks a sample size that balances the statistical gain from additional data against its acquisition cost. For a fixed sample size \(n\), the corresponding risk is
\begin{equation}
R_n^{\mathrm{fixed}}(\hat{k}_n)
=
\frac{A\,\mathbb{E}[V_n(\hat{k}_n)]}{n}
+cn.
\label{fixed_risk_function}
\end{equation}
In a sequential procedure, however, the sample size at which estimation is performed is generally data-dependent and hence random. To accommodate this additional randomness, for a generic random sample size \(M\) we define
\begin{equation}
R_M(\hat{k}_M)
=
A\,\mathbb{E}\!\left[\frac{V_M(\hat{k}_M)}{M}\right]
+c\,\mathbb{E}[M].
\label{risk_function}
\end{equation}
When \(M=n\) is non-random, this reduces to the fixed-sample risk in \eqref{fixed_risk_function}. The corresponding population risk at a fixed sample size \(n\) is
\begin{equation}\label{oracle_risk}
R_n^\star
=
\frac{A\xi_n}{n}+cn,
\end{equation}
where \(\xi_n\) is the population residual variance defined in \eqref{V_Xi}. This formulation will later be used to evaluate the data-dependent sample sizes generated by the proposed sequential procedure (see Section \ref{sec:asym_efficiency} and \ref{sec_simu}). The oracle sample size is then the smallest sample size minimizing this population risk,
\begin{equation}
\label{optimum_sample_size}
n_0(c)
:=
\min\left\{
n\in\mathbb N:
n\in\arg\min_{r\in\mathbb N}R_r^\star
\right\}.
\end{equation}
Thus, \(n_0(c)\) provides the population benchmark against which the sample sizes generated by the sequential procedure will be compared. However, the definition alone does not guarantee that such a minimizer exists. In the fixed-dimensional setting considered by \cite{BanerjeeChattopadhyay2025}, the residual variance \(\xi\) does not depend on \(n\), so the criterion \(A\xi/n+cn\) is convex in \(n\) and yields the familiar oracle scale \(n_0(c)\asymp\sqrt{A\xi/c}\), with the integer optimum obtained by discretizing the continuous solution. In our high-dimensional setting, both \(p_n\) and \(\Sigma_{0n}\) may vary with \(n\), and consequently \(\xi_n\) also depends on \(n\). Hence, no analogous closed-form solution is available in general. We therefore need to characterize \(n_0(c)\) more carefully, establish its existence, and derive the properties required for the subsequent sequential analysis; this is done in the next section.
\subsection{Oracle characterization}\label{oracle_char}

\noindent In this section, we characterize the oracle sample size \(n_0(c)\) and establish several properties needed for the subsequent analysis. We first introduce the assumptions imposed on the population-level data-generating mechanism that will be used throughout the paper.

\begin{enumerate}[label=(\textbf{AN\arabic*})]

\item \label{as:spec_subg}
There exists a constant $\kappa_\sigma \in (0,1]$, not depending on $n$, such
that
$\kappa_\sigma \leq \lambda_1(\Sigma_{0n}) \leq \kappa_\sigma^{-1}.$
That is, the maximum eigenvalue of $\Sigma_0$ is bounded. Second, let  $X_1 = \Sigma_{0n}^{1/2} Y_1$
where $Y_1$ is a isotropic sub-Gaussian random variable. Then, there is a constant $\sigma_0 > 0$, not depending on $n$, such that $\|Y_1\|_{\psi_2} \le \sigma_0$.

\item \label{ass:non_incr_tail_mass}
The sequence $\{\xi_n\}_{n\geq1}$ is positive and nonincreasing, that is, $\xi_n>0$ and $\xi_{n+1}\leq \xi_n$ for every $n\in\mathbb{N}$, with $\inf_{n\geq1}\xi_n>0$.
\item \label{ass:oracle_flatness}
As $c \downarrow 0$, $\xi_{n_0(c)-1}/\xi_{n_0(c)}
=
1 + o\left(1/n_0(c)\right).$

\item \label{ass:gap_k0}
Fix $\eta \in (0,1)$ and let $\gamma_k$ be defined as in Section \ref{setup_target}. Assume
\[
g_0
:=
\inf_n
\min\left\{
\eta-\gamma_{k_{0n}-1},
\gamma_{k_{0n}}-\eta
\right\}
>0.
\]

\end{enumerate}
The uniform upper bound on the largest eigenvalue is standard in high-dimensional covariance estimation; see, e.g., \cite{banerjee1,banerjee2,bickel2008regularized,spectrum,xiang,sarkar1,sarkar2}. The sub-Gaussian condition in Assumption~\ref{as:spec_subg} rules out heavy tails and holds automatically in the Gaussian case after whitening.

Assumption~\ref{ass:non_incr_tail_mass} ensures that the population residual variance \(\xi_n\) does not fluctuate upward with \(n\), preserving the trade-off between the decreasing term \(A\xi_n/n\) and the increasing sampling cost \(cn\). In the fixed-dimensional setting \citep[e.g.,][]{chaban2022} \(\xi_n\) is constant and this trade-off is automatic, whereas in the growing-dimensional regime it is needed to ensure a well-behaved oracle problem. Assumption~\ref{ass:oracle_flatness} further controls the local variation of \(\xi_n\) near \(n_0(c)\), allowing the discrete oracle condition to be approximated by the surrogate characterization in Lemma~\ref{lem:asymp_equi}.

Finally, Assumption~\ref{ass:gap_k0} stabilizes the selected population dimension by separating the threshold \(\eta\) from \(\gamma_{k_{0n}-1}\) and \(\gamma_{k_{0n}}\), thereby preventing small eigenvalue perturbations from changing \(k_{0n}\). Since \(\lambda_{k_{0n}}(\Sigma_0)/\operatorname{tr}(\Sigma_0)=\gamma_{k_{0n}}-\gamma_{k_{0n}-1}\ge2g_0\), it also implies \(\sup_n\widetilde r(\Sigma_{0n})<\infty\). Thus, while the ambient dimension \(p_n\) may diverge, the effective dimension remains uniformly controlled; in this sense, it also acts as a growth condition for the high-dimensional regime.

\noindent With these assumptions in place, we next show that the oracle sample size \(n_0(c)\) is well defined and has the desired asymptotic behavior. The formal proof is provided in Appendix~\ref{app_A}.

\begin{lemma}
\label{lem:oracle_optimality}
Assume \(A>0\), \(c>0\), and Assumption~\ref{ass:non_incr_tail_mass} holds. Then the set of global minimizers of \(R_n^\star\) over \(n\in\mathbb N\) is nonempty, and \(n_0(c)\), defined in \eqref{optimum_sample_size} as the smallest such minimizer, is well defined and unique. Moreover, \(n_0(c)\to\infty\) as \(c\downarrow0\).
\end{lemma}
\noindent Although Lemma~\ref{lem:oracle_optimality} shows that \(n_0(c)\) is well defined, it is characterized only implicitly as the minimizer of the population risk \(R_n^\star\) and depends on the unknown tail-mass sequence \(\{\xi_n\}\). This makes it analytically inconvenient for subsequent comparisons. We therefore introduce the auxiliary index
\begin{equation}
\label{auxiliary_index}
n_T(c)
:=
\min\left\{
n\in\mathbb N:
\frac{A\xi_n}{n^2}\le c
\right\}.
\end{equation}
The definition of $n_T(c)$ admits a natural decision-theoretic
interpretation. The risk $R_n^\star$ defined in \eqref{oracle_risk} is composed of two competing
components: the first term corresponds to the estimation error, which
decreases with $n$, whereas the second term represents a linear sampling
cost, which increases with $n$. For small sample sizes, the estimation error
dominates, so that enlarging $n$ reduces the overall risk. In contrast, for
sufficiently large $n$, the linear cost term prevails, and further sampling
leads to an increase in risk. The optimal sample size is therefore determined
by a balance between these two opposing terms. We next describe conditions
under which $n_T(c)$ and $n_0(c)$ are asymptotically equivalent.

Since $n_0(c)$ is defined as the minimizer of 
$R_n^\star = A\xi_n/n + cn$, the relevant first-order characterization is
governed by discrete increments rather than an ordinary derivative of a
continuous function. In particular, at the optimum $n_0= n_0(c)$, the
decrease in the estimation component must be of the same order as the
increase in the linear cost $cn$. Accordingly, the quantity
$A\xi_{n_0-1}/(n_0-1) - A\xi_{n_0}/n_0$ plays the role of a discrete
derivative. A direct decomposition yields
\[
\frac{\xi_{n-1}}{n-1} - \frac{\xi_n}{n} = \frac{\xi_n}{n(n-1)} + \frac{\xi_{n-1} - \xi_n}{n-1}.
\]
The first term is of order $\xi_n/n^2$, which coincides with the scale
appearing in the auxiliary index $n_T(c)$. Therefore, for the discrete
optimality condition at $n_0(c)$ to imply $c = A\xi_{n_0(c)}/n_0(c)^2\,\{1 +
o(1)\}$, it is necessary that the second term be asymptotically negligible
relative to the first. This requirement is ensured by the local smoothness
condition stated as Assumption~\ref{ass:oracle_flatness},
namely $\xi_{n_0(c)-1}/\xi_{n_0(c)} = 1 + o\{1/n_0(c)\}$. From the
discrete-derivative perspective, this condition asserts that the relative
one-step variation of $\xi_n$ is asymptotically smaller than the intrinsic
variation of the factor $1/n$, whose own discrete derivative is of order
$1/n^2$. Under this regime, the dominant contribution to the local slope of
$n \mapsto \xi_n/n$ arises from the denominator $n^{-1}$, while the numerator
$\xi_n$ varies sufficiently slowly so as not to affect the leading-order
balance. This justifies the asymptotic identification of $n_0(c)$ through the auxiliary index $n_T(c)$. We next state this condition formally and present a lemma showing that $n_T(c)/n_0(c)\to 1$ as $c\downarrow 0$. The proof of this lemma is provided in Appendix~\ref{app_A}.

\begin{lemma}\label{lem:asymp_equi}
Assume \(A>0\), \(c>0\), and that Assumptions~\ref{ass:non_incr_tail_mass} and~\ref{ass:oracle_flatness} hold. Recall the oracle index \(n_0(c)\) defined in \eqref{optimum_sample_size} and the auxiliary index \(n_T(c)\) defined in \eqref{auxiliary_index}. Then, as \(c\downarrow0\), \(c=A\xi_{n_0(c)}n_0(c)^{-2}\{1+o(1)\}\), and consequently \(n_T(c)/n_0(c)\to1\).
\end{lemma}

\section{Three-Stage PCA Procedure}
\label{sec:three_stage}

\noindent The oracle characterization in Section~\ref{oracle_char} suggests a natural sequential strategy. In particular, Lemma~\ref{lem:asymp_equi} shows that the oracle sample size asymptotically satisfies \(n_0(c)\simeq\sqrt{A\xi_{n_0(c)}/c}\). The difficulty is that the population residual variance \(\xi_{n_0(c)}\), and hence \(n_0(c)\) itself, is unknown. A natural approach is therefore to begin with a pilot sample, estimate the residual variance from the available data, and use this plug-in estimate to determine how many additional observations should be collected. More generally, this idea belongs to the classical sequential and multistage sampling literature, where an initial sample is used to estimate an unknown quantity governing the optimal sample size and later sampling decisions are updated using the information accumulated so far; see, for example, \cite{hall1981asymptotic,ghosh201997sequential,mukhopadhyay2009sequential}.

The simplest implementation uses the pilot sample once to estimate the entire target sample size and then terminates after collecting the prescribed additional observations. Such a procedure has two sampling stages: the pilot stage and the final stage. A two-stage version of this plug-in idea was introduced for cost-compression PCA by \cite{chaban2022}. Their formulation, however, treats the ambient dimension and the number of retained principal components \(k\) as fixed, so that the population residual variance does not vary with the sample size, and the resulting theory primarily concerns first-order efficiency. A subsequent modified two-stage procedure obtains second-order guarantees under stronger assumptions, including multivariate normality and knowledge of the smallest population eigenvalue, or a positive lower bound for it \citep[e.g.,][]{BanerjeeChattopadhyay2025}. In contrast, our setting allows \(p_n\) to grow with \(n\), selects the retained dimension through the explained-variance threshold \(\eta\in(0,1)\), and does not require prior knowledge of population eigenvalues.

These features motivate an additional recalibration before sampling terminates. Rather than using the pilot estimate to determine the entire final sample size, we first use it to move only to an intermediate sample size, chosen to be approximately a fraction \(\rho\in(0,1)\) of the oracle target. The additional observations collected at this intermediate stage provide an updated covariance matrix and residual-variance estimate, which are then used to determine the final sample size. We therefore obtain a three-stage procedure: a pilot stage, an intermediate recalibration stage, and a final sampling stage. In this way, errors in the pilot estimate are not carried directly into the final sampling decision. We take
\begin{equation}
\label{pilot}
m=\left\lceil\left(\frac{A}{c}\right)^{1/(2\delta)}\right\rceil,
\qquad \delta>1.
\end{equation}
This choice allows the pilot size to increase as the sampling cost \(c\) decreases while remaining asymptotically smaller than the target sample size. Thus, the pilot provides enough information to estimate the relevant PCA quantities without consuming a substantial fraction of the eventual sample. The parameter \(\delta\) controls this trade-off: values closer to \(1\) produce a larger pilot, whereas larger values produce a smaller initial sample. The precise asymptotic requirements on \(m\) are formalized through Assumptions~\ref{as:A1} below. 

\noindent We now formally present the proposed three-stage procedure. For a fixed explained-variance threshold \(\eta\in(0,1)\), the retained PCA dimension \(\hat k_n\) at each stage is chosen as the smallest number of sample principal components explaining at least a fraction \(\eta\) of the total sample variation. Using the pilot residual \(V_m(\hat k_m)\), define the intermediate sample size
\begin{equation}
\label{3stageT}
T_m
:=
\max\left\{
m,\,
\left\lceil
\rho\sqrt{\frac{A}{c}V_m(\hat k_m)}
\right\rceil
\right\},
\qquad \rho\in(0,1),
\end{equation}
and write \(T=T_m\). After collecting observations up to \(T\), PCA is recomputed using the same explained-variance threshold \(\eta\), and the updated residual \(V_T(\hat k_T)\) is used to determine the final sample size
\begin{equation}
\label{3stageN}
N_T
:=
\max\left\{
T,\,
\left\lceil
\sqrt{\frac{A}{c}V_T(\hat k_T)}
\right\rceil
\right\}.
\end{equation}
The factor \(\rho\) keeps the intermediate stage below the estimated oracle target, leaving room for a final recalibration based on the more informative sample of size \(T\). The complete procedure is summarized in Algorithm~\ref{algo_1}.

\begin{algorithm}[!htbp]
\caption{Three-Stage PCA Procedure}
\label{algo_1}
\begin{algorithmic}[1]
\Require $A>0$,\ $c>0$,\ $\rho\in(0,1)$,\ $\delta>1$,\ $\eta\in(0,1)$
\Ensure Compressed dataset $\mathcal{D}_{N_T}^{\hat{k}_{N_T}}$

\Statex
\Statex \textbf{Stage 1: Pilot sample}
\State Compute the pilot sample size
       $m \gets \left(\dfrac{A}{c}\right)^{\!1/(2\delta)}$.
\State Collect $m$ observations and form
       $\mathcal{D}_m^{p}=(X_1,X_2,\ldots,X_m)$.
\State Apply PCA on $\mathcal{D}_m^{p}$.
\State Determine $\hat{k}_m$, the minimum number of PCs explaining
       at least $100\eta\%$ of the variance.
\State Compute
       $V_m(\hat{k}_m)\gets\operatorname{tr}\!\left[(I-P_m^{\hat{k}_m})S_m\right]$.

\Statex
\Statex \textbf{Stage 2: First update}
\State Compute
       \[
          T_m \gets \max\!\left\{\,m,\;
          \left\lceil \rho\sqrt{\tfrac{A}{c}\,V_m(\hat{k}_m)}\,\right\rceil\right\},
          \qquad T \gets T_m.
       \]
\State Collect $T_m-m$ additional observations and form
       $\mathcal{D}_T^{p}=(X_1,X_2,\ldots,X_T)$.
\State Apply PCA on $\mathcal{D}_T^{p}$.
\State Determine $\hat{k}_T$, the minimum number of PCs explaining
       at least $100\eta\%$ of the variance.
\State Compute
       $V_T(\hat{k}_T)\gets\operatorname{tr}\!\left[(I-P_T^{\hat{k}_T})S_T\right]$.

\Statex
\Statex \textbf{Stage 3: Final sample}
\State Compute
       \[
          N_T \gets \max\!\left\{\,T_m,\;
          \left\lceil \sqrt{\tfrac{A}{c}\,V_T(\hat{k}_T)}\,\right\rceil\right\}.
       \]
\State Collect $N_T-T_m$ additional observations and form
       $\mathcal{D}_{N_T}^{p}=(X_1,X_2,\ldots,X_{N_T})$.
\State Apply PCA on $\mathcal{D}_{N_T}^{p}$ to obtain $\hat{k}_{N_T}$.
\State \textbf{Return} the compressed dataset \(\mathcal{D}_{N_T}^{\hat{k}_{N_T}}\)
\end{algorithmic}
\end{algorithm}

An essential requirement for any sequential sampling rule is that it terminates after finitely many observations. The following result guarantees this property for Algorithm~\ref{algo_1}; its proof is provided in Appendix~\ref{app_A}.

\begin{lemma}
\label{lem:sampleext}
For any \(c>0\), the stopping sample size \(N_T\) in Algorithm~\ref{algo_1} is almost surely finite, that is, \(P(N_T<\infty)=1\).
\end{lemma}

\paragraph{Two-stage comparator.}
For comparison, we also consider the natural two-stage version of the procedure. Setting \(\rho=1\) in \eqref{3stageT} gives
\[
T_m^{(2)}
=
\max\left\{
m,\,
\left\lceil
\sqrt{\frac{A}{c}V_m(\hat k_m)}
\right\rceil
\right\}.
\]
The two-stage procedure terminates at \(T_m^{(2)}\), without performing the final recalibration in \eqref{3stageN}. Thus, its final sample-size decision is based entirely on the pilot estimate \(V_m(\hat k_m)\), whereas the proposed three-stage procedure uses the pilot only to reach an intermediate fraction of the target and then updates the residual variance before making the final decision. In this sense, the third stage is not merely a technical refinement: the additional recalibration controls the effect of pilot-stage estimation error at the second-order level. This is particularly useful in the present high-dimensional setting, where both the covariance structure and the selected PCA dimension may evolve with \(n\). As shown in Section~\ref{second_order}, the additional recalibration leads to sharper approximations to both the oracle sample size and risk and mitigates the deterioration that can arise from relying solely on the pilot estimate.
\paragraph{Implementation remarks.}
The choice of \(\rho\) controls how aggressively the procedure moves toward the estimated target after the pilot stage. A value of \(\rho\) close to one collects more observations before recalibration and therefore provides a more accurate Stage-2 estimate, but leaves less room for the final adjustment; smaller values allow earlier recalibration but use fewer observations to estimate \(V_T(\hat k_T)\). In practice, \(\rho\) can therefore be chosen as a moderate fixed fraction of the estimated target. For example, we use \(\rho=0.7\) in both the simulation study (Section~\ref{sec_simu}) and the real-data analysis (Section~\ref{application}). The theoretical results below allow any fixed \(\rho\in(0,1)\). Integer rounding in \eqref{pilot}--\eqref{3stageN} has no effect on the asymptotic results.

\section{Asymptotic Efficiency of the Three-stage PCA Procedure}\label{sec:asym_efficiency}

\noindent In this section, we study the efficiency properties of the proposed three-stage procedure in the high-dimensional setting. We first establish first-order efficiency in terms of sample-size consistency, expected sample size, and risk. We then investigate second-order guarantees under additional regularity conditions. Finally, we compare these results with the two-stage procedure introduced in Section~\ref{sec:three_stage} and clarify the benefit of the additional recalibration stage.
\subsection{First-order efficiency}
\noindent The goal of first-order efficiency is to show that, asymptotically, the proposed data-driven procedure behaves like the oracle rule. In particular, we want the selected sample size to be close to the oracle target \(n_0(c)\), not only along individual sample paths but also on average, and we want the resulting risk to be asymptotically equivalent to the oracle risk. Thus, the first-order analysis asks whether estimating the unknown residual variance sequentially leads to any asymptotic loss relative to knowing the oracle target in advance. Along with Assumptions~\ref{as:spec_subg}--\ref{ass:gap_k0} introduced earlier, the first-order analysis requires two additional conditions governing the pilot sample size and its relation to the oracle target.

\vspace{0.2cm}
\noindent\textbf{Regularity conditions}

\begin{enumerate}[label=(A\arabic*)]
\item \label{as:A1} \(m=o\bigl(n_0(c)\bigr)\) as \(c\downarrow0\).

\item \label{as:A2} \(\xi_m=\xi_{n_0(c)}\{1+o(1)\}\) as \(c\downarrow0\).
\end{enumerate}

\noindent Assumption~\ref{as:A1} ensures that the pilot sample does not overtake the oracle target, so that the subsequent recalibrations remain meaningful. Under the pilot choice (see \eqref{pilot}) in Algorithm~\ref{algo_1}  this condition holds automatically and hence imposes no additional restriction. Assumption~\ref{as:A2} ensures that the pilot-stage residual variance \(\xi_m\) is representative of its value at the oracle index, preventing systematic bias in the initial sample-size estimate. Under these conditions, the final sample size is strongly consistent for the oracle target.

\begin{theorem}\label{th:th_almost_sure_cons}
Consider Algorithm~\ref{algo_1} with \(A>0\) and \(\rho\in(0,1)\), and let \(\mathbb P_0\) denote the true data-generating mechanism. Under Assumptions~\ref{as:spec_subg}--\ref{ass:gap_k0} and~\ref{as:A1}--\ref{as:A2}, \(N_T/n_0(c)\xrightarrow{\mathrm{a.s.}}1\) as \(c\downarrow0\) under \(\mathbb P_0\).
\end{theorem}

\noindent The detailed proof is provided in Appendix~\ref{app_A}. Theorem~\ref{th:th_almost_sure_cons} establishes almost sure convergence of \(N_T/n_0(c)\), but this alone does not imply convergence of \(\mathbb E_0[N_T]/n_0(c)\) or of the risk ratio \(R_{N_T}(\hat k_{N_T})/R_{n_0(c)}\). In particular, rare events on which the stopping sample size deviates substantially from \(n_0(c)\) may have vanishing probability but still contribute non-negligibly to expectations. To rule out such tail effects, we impose an additional condition controlling \(\xi_n\) uniformly in a neighborhood of \(n_0(c)\) and its behavior relative to the pilot scale.

\vspace{0.2cm}
\noindent\textbf{Regularity condition}

\begin{enumerate}[label=(B\arabic*)]
\item \label{as:B1}
There exists a constant \(\alpha\in(1-\rho,1)\), independent of \(m\) and \(c\), such that, as \(c\downarrow0\),
\(\xi_{\lfloor(1-\alpha)n_0(c)\rfloor}
=
\xi_{n_0(c)}\{1+o(1)\}\) and
\(\xi_{\lceil(1+\alpha)n_0(c)\rceil}
=
\xi_{n_0(c)}\{1+o(1)\}\).
Moreover, \(\sup_{t\ge1}\xi_t/\xi_{n_0(c)}=O(1)\).
\end{enumerate}

\noindent Assumption~\ref{as:B1} ensures that the residual variance remains stable in a neighborhood of the oracle sample size and that its relative magnitude does not become unbounded. Here \(\rho\in(0,1)\), introduced in Algorithm~\ref{algo_1}, determines the fraction of the oracle target reached at the intermediate stage, while \(\alpha\) determines the neighborhood over which stability of \(\xi_n\) is required. The restriction \(\alpha>1-\rho\) ensures that this neighborhood contains the intermediate target \(\rho n_0(c)\). Consequently, as expected, a smaller \(\rho\) requires stability over a wider range of sample sizes and makes Assumption~\ref{as:B1} more restrictive. Under this condition, almost sure consistency can be strengthened to mean sample-size and risk efficiency.

\begin{theorem}\label{th:mean_cons}
Under the conditions of Theorem~\ref{th:th_almost_sure_cons} and Assumption~\ref{as:B1}, \(\mathbb E_0[N_T]/n_0(c)\to1\) as \(c\downarrow0\).
\end{theorem}

\begin{theorem}\label{th:risk_cons}
Under the conditions of Theorem~\ref{th:mean_cons}, \(R_{N_T}(\hat k_{N_T})/R_{n_0(c)}\to1\) as \(c\downarrow0\).
\end{theorem}

\noindent The proofs of Theorem \ref{th:mean_cons} and Theorem \ref{th:risk_cons} are provided in the Appendix \ref{app_A}. 
Theorems~\ref{th:th_almost_sure_cons}--\ref{th:risk_cons} establish the 
first-order, or ratio, efficiency of the three-stage procedure, showing that
the sample size and risk of Algorithm~\ref{algo_1}, when normalized by their
respective oracle quantities, converge to one. Such ratio consistency,
however, does not quantify the corresponding differences on the original
scale, which may remain non-negligible even when the ratios converge to one.
This motivates the study of second-order efficiency, which provides finer
control of the deviations from the oracle procedure.

\subsection{Second-order efficiency}\label{second_order}

\noindent We now turn to second-order efficiency, which provides a finer comparison with the oracle rule than the ratio-based first-order results. While first-order efficiency establishes asymptotic equivalence after normalization by the oracle sample size or oracle risk, second-order efficiency examines the corresponding unnormalized deviations. In particular, we study how closely the expected stopping sample size tracks \(n_0(c)\) itself and how closely the resulting risk tracks the oracle risk. Fluctuations that are negligible at the first-order scale may still contribute to these differences, so a more refined analysis is needed. This requires stronger rate conditions on the local behavior of the residual variance and additional stability of the principal subspace around the PCA cutoff.

\vspace{0.2cm}
\noindent\textbf{Additional regularity conditions}

\begin{enumerate}[label=(BN\arabic*)]
\item \textbf{Rate version of oracle flatness}\label{ass:oracle_flatness_rate}
As \(c\downarrow0\), \(\xi_{n_0(c)-1}/\xi_{n_0(c)}=1+O\{n_0(c)^{-2}\}\).

\item \textbf{Spectral separation at the PCA cutoff}\label{ass:spectral_gap_k0}
In addition to the threshold gap condition in Assumption~\ref{ass:gap_k0}, assume that there exists a constant \(\gamma>0\), independent of \(n\), such that \(\lambda_{k_{0n}}(\Sigma_0)-\lambda_{k_{0n}+1}(\Sigma_0)\ge\gamma\).
\end{enumerate}

\begin{enumerate}[label=(CN\arabic*)]
\item \textbf{Rate version of Assumption~\ref{as:B1}}\label{ass:B1A}
There exists a constant \(\alpha\in(1-\rho,1)\), independent of \(m\) and \(c\), such that, as \(c\downarrow0\), \(\xi_{\lfloor(1-\alpha)n_0(c)\rfloor}=\xi_{n_0(c)}\{1+O(n_0(c)^{-1})\}\) and \(\xi_{\lceil(1+\alpha)n_0(c)\rceil}=\xi_{n_0(c)}\{1+O(n_0(c)^{-1})\}\). Moreover, \(\sup_{t\ge1}\xi_t/\xi_{n_0(c)}=O(1)\).
\end{enumerate}

\noindent Assumptions~\ref{ass:oracle_flatness_rate} and~\ref{ass:B1A} are rate-strengthened versions of Assumptions~\ref{ass:oracle_flatness} and~\ref{as:B1}, respectively. The former provides a sharper local approximation to the oracle criterion, while the latter controls the variation of the residual variance over the range of sample sizes relevant to the sequential procedure. As before, the restriction \(\alpha>1-\rho\) ensures that this range contains the intermediate target \(\rho n_0(c)\), so smaller values of \(\rho\) require stability over a wider neighborhood.
\begin{remark}
Under Assumption~\ref{ass:oracle_flatness_rate}, together with Assumptions~\ref{as:spec_subg} and~\ref{ass:non_incr_tail_mass}, Lemma~\ref{lem:asymp_equi} can be sharpened to \(c=A\xi_{n_0(c)}n_0(c)^{-2}\{1+O(n_0(c)^{-1})\}\), while \(n_T(c)/n_0(c)\to1\) continues to hold. The proof follows the same argument as Lemma~\ref{lem:asymp_equi}, replacing the corresponding \(o(1)\) terms by their rate bounds, and is therefore omitted.
\end{remark}

\noindent The threshold gap in Assumption~\ref{ass:gap_k0} and the spectral gap in Assumption~\ref{ass:spectral_gap_k0} play distinct roles. The former stabilizes the selected dimension under the explained-variance rule, preventing small perturbations from changing \(k_{0n}\). The latter stabilizes the corresponding principal subspace through a Davis--Kahan type argument \citep[e.g.,][]{davis1,davis2}, allowing the difference between the empirical and population projection matrices to be controlled by the covariance estimation error. This stronger projection stability is particularly important for establishing the signed bias bound in Lemma~\ref{lem:sign_bias_v}, which is a key component to proof Theorems~\ref{th:mean_diff} and~\ref{th:risk_diff}. With these assumptions in hand, we now present the main second-order efficiency results for the proposed three-stage procedure.

\begin{theorem}\label{th:mean_diff}
Consider Algorithm~\ref{algo_1} with \(A>0\) and \(\rho\in(0,1)\), and let \(\mathbb P_0\) denote the true data-generating mechanism. Suppose Assumptions~\ref{as:spec_subg}, \ref{ass:non_incr_tail_mass}, \ref{ass:gap_k0}, \ref{ass:oracle_flatness_rate}, \ref{ass:spectral_gap_k0}, \ref{as:A1}--\ref{as:A2}, and~\ref{ass:B1A} hold. Then, as \(c\downarrow0\), \(\left|\mathbb E_0[N_T]-n_0(c)\right|=O(1)\).
\end{theorem}

\noindent The corresponding result for the risk shows that the data-driven procedure also remains close to the oracle at the second-order scale.

\begin{theorem}\label{th:risk_diff}
Under the same conditions as in Theorem~\ref{th:mean_diff}, as \(c\downarrow0\), \(\left|R_{N_T}(\hat k_{N_T})-R_{n_0(c)}\right|=O(c)\), where \(R_{N_T}(\hat k_{N_T})\) is defined in terms of \(\mathbb E_0\), the expectation under \(\mathbb P_0\).
\end{theorem}

\subsection{Comparison with two-stage sampling}\label{two_stage_comparison}

\noindent We finally compare the second-order guarantees above with those of the two-stage procedure introduced in Section~\ref{sec:three_stage}. Recall that the two-stage rule uses the pilot sample to determine the final sample size \(T_m^{(2)}\) directly, without the intermediate recalibration employed by Algorithm~\ref{algo_1}. Equivalently, it corresponds to using the intermediate rule with \(\rho=1\) and stopping after that stage.

\noindent Following the same argument as in the proof of Theorem~\ref{th:mean_diff}, the corresponding two-stage sample size satisfies
\[
\left|\mathbb E_0[T_m^{(2)}]-n_0(c)\right|
=
O\left(\frac{n_0(c)}{m}\right).
\]
Since \(m=o(n_0(c))\), the factor \(n_0(c)/m\) diverges, so this bound need not remain bounded, whereas the proposed three-stage procedure achieves the \(O(1)\) deviation established in Theorem~\ref{th:mean_diff}. Similarly, the corresponding two-stage risk deviation is of order $O\left(c\,n_0(c)/m\right),$ whereas Theorem~\ref{th:risk_diff} gives the sharper \(O(c)\) bound for the three-stage procedure.

\noindent Thus, the additional stage is substantive rather than merely a technical device in the high-dimensional setting considered here. In the three-stage procedure, the pilot sample is used only to move the sampling process toward a fraction of the oracle target, after which the additional observations collected at the intermediate stage are used to update the residual-variance estimate before the final sampling decision. This recalibration reduces the influence of pilot-stage estimation error on the final sample size and removes the additional factor \(n_0(c)/m\) from both the sample-size and risk deviation bounds. The comparison therefore highlights the role of the third stage in achieving sharper second-order efficiency in the high-dimensional regime.

\section{Simulation Study}\label{sec_simu}
\noindent \textit{Simulation Setup.} In this section, we conduct a simulation study to provide empirical evidence for the first- and second-order efficiency results of the proposed three-stage PCA procedure. We consider \(p\in\{200,400,800\}\) and, for each \(p\), vary the oracle sample size over \(n_0\in\{50,100,\ldots,1050\}\). We set \(A=1,\eta=0.60,\rho=0.70,\delta=1.20\).

For each value of \(p\), we generate three dense covariance matrices of the form \(\Sigma=Q\Lambda Q^T\), where \(Q\) is a randomly generated \(p\times p\) orthogonal matrix. The same \(Q\) is used for all three covariance structures at a given \(p\), so that the models differ only through their eigenvalue decay patterns. Writing \(q=p-3\), we consider three specifications for \(\Lambda\): (a) the \textit{uniform-tail model}, with \(\Lambda_{\mathrm U}=\operatorname{diag}(4,2,1,1/q,\ldots,1/q)\); (b) the \textit{polynomial-tail model}, with \(\Lambda_{\mathrm P}=\operatorname{diag}(4,2,1,w_1^{(P)},\ldots,w_q^{(P)})\), where \(w_j^{(P)}=j^{-2}/\sum_{\ell=1}^q \ell^{-2}\); and (c) the \textit{exponential-tail model}, with \(\Lambda_{\mathrm E}=\operatorname{diag}(4,2,1,w_1^{(E)},\ldots,w_q^{(E)})\), where \(w_j^{(E)}=0.5^{j-1}/\sum_{\ell=1}^q 0.5^{\ell-1}\).

Note that, for all three models, the first two cumulative explained-variance proportions are \(\gamma_1=4/8=0.50\) and \(\gamma_2=6/8=0.75\). Hence, with \(\eta=0.60\), the true number of principal components required to explain at least \(100\eta\%\) of the total variance is \(k_{0n}=2\) for every model and every value of \(p\). The corresponding residual variance is \(\xi_n=\sum_{j>2}\lambda_j(\Sigma)=2\) for every model and every value of \(p\). Moreover, the threshold \(\eta\) is uniformly separated from both \(\gamma_1\) and \(\gamma_2\), and the eigengap at \(k_{0n}\) is \(\lambda_2-\lambda_3=1\). Thus, these covariance models satisfy all the necessary assumptions required by our theoretical results. Since \(A=1\) and \(\xi_n=2\), the sampling cost corresponding to the oracle sample size \(n_0\) is \(c=2/n_0^2\), and with \(\delta=1.20\), the pilot sample size is \(m=\left\lceil (n_0^2/2)^{5/12}\right\rceil\). The corresponding population oracle risk is directly available as \(R_{n_0}=A\xi_n/n_0+cn_0=4/n_0\).

We consider \(5000\) Monte Carlo replications. For each Monte Carlo replication, observations are independently generated from \(N_p(0,\Sigma)\). For each dataset, starting from the pilot sample size \(m\), we apply Algorithm~\ref{algo_1} and compute the final sample size \(N_T\), the selected dimension \(\hat{k}_{N_T}\), and the corresponding residual variance \(V_{N_T}\). Specifically, let \(N_T^{(b)}\), \(\hat{k}_{N_T}^{(b)}\), and \(V_{N_T}^{(b)}\) denote the corresponding quantities from the $b^\text{th}$ Monte Carlo replication, \(b=1,\ldots,5000\). We estimate \(\mathbb E(N_T)\) by \(\widehat{\mathbb E}(N_T)=B^{-1}\sum_{b=1}^B N_T^{(b)}\), the relative-error probability by \(\widehat{P}(|N_T/n_0-1|>0.05)=B^{-1}\sum_{b=1}^B I(|N_T^{(b)}/n_0-1|>0.05)\), and the sequential risk by \(\widehat{R}_{N_T}=B^{-1}\sum_{b=1}^B\{A V_{N_T}^{(b)}(\hat{k}_{N_T}^{(b)})/N_T^{(b)}+cN_T^{(b)}\}\). We also estimate the probability of correctly selecting the true number of principal components by \(\widehat{P}(\hat{k}_{N_T}=k_{0n})=B^{-1}\sum_{b=1}^B I(\hat{k}_{N_T}^{(b)}=k_{0n})\).
For each \(p\), we plot six quantities as functions of \(n_0\) in Figure~\ref{fig:simulation_results}: the relative-error probability, the first-order sample-size ratio, the first-order risk ratio, the second-order sample-size error, the second-order risk regret, and the probability of correctly selecting \(k_{0n}\). 

\vspace{0.2cm}
\noindent \textit{Results.} Figures~\ref{fig:sim_p200}--\ref{fig:sim_p800} in Figure~\ref{fig:simulation_results} provide strong empirical support for the theoretical findings in Section \ref{sec:asym_efficiency}. Across all three covariance models and all values of \(p\), the results demonstrate first-order efficiency in both sample size and risk as \(n_0\) increases. The second-order sample-size error remains well controlled, while the risk regret decreases rapidly toward zero. The probability of correctly selecting the population dimension also approaches one quickly. Overall, the performance remains stable as the ambient dimension increases from \(p=200\) to \(p=800\).

\begin{figure}[!htbp]
    \centering
    \begin{subfigure}[t]{0.48\textwidth}
        \centering
        \includegraphics[width=\textwidth]{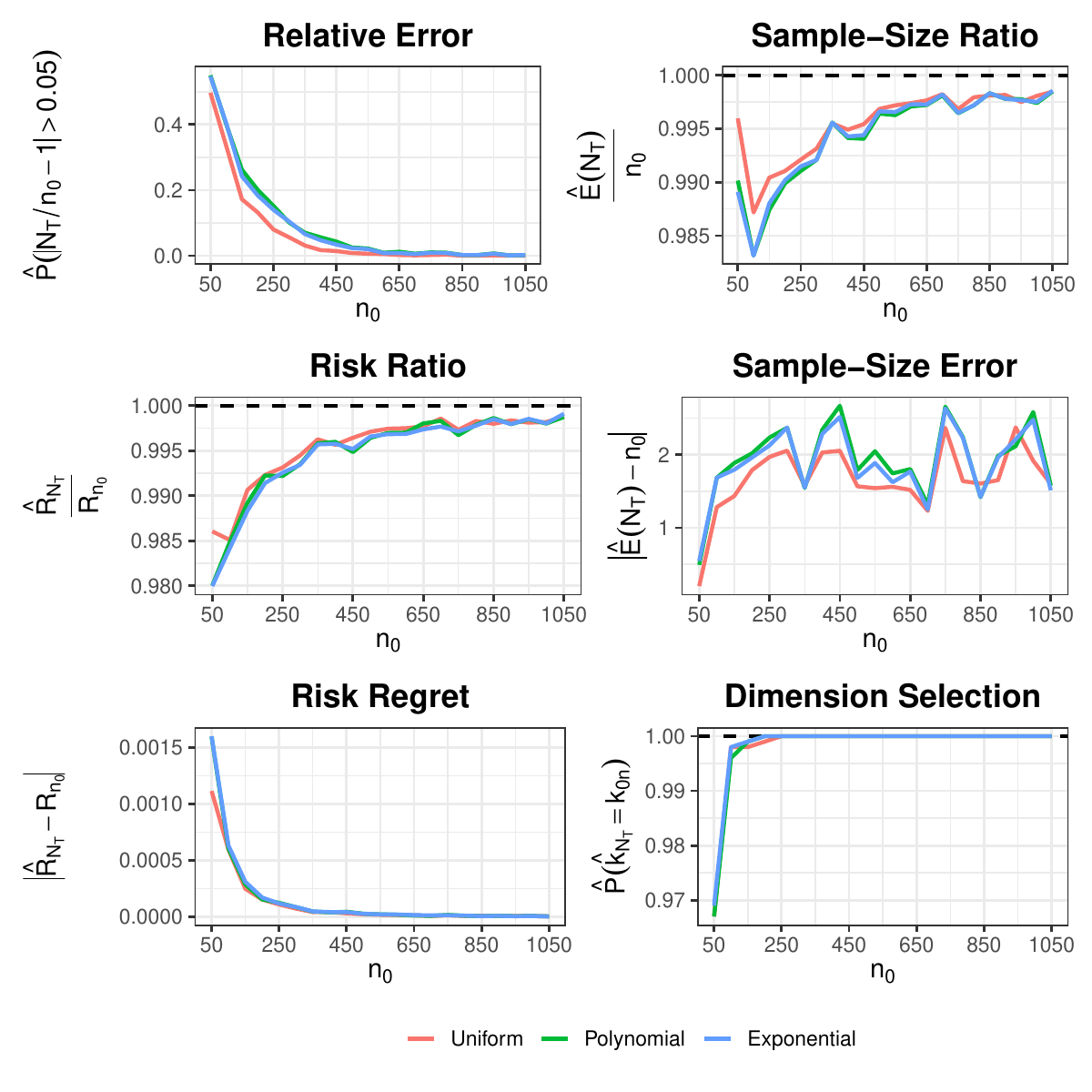}
        \caption{\(p=200\)}
        \label{fig:sim_p200}
    \end{subfigure}
    \hfill
    \begin{subfigure}[t]{0.48\textwidth}
        \centering
        \includegraphics[width=\textwidth]{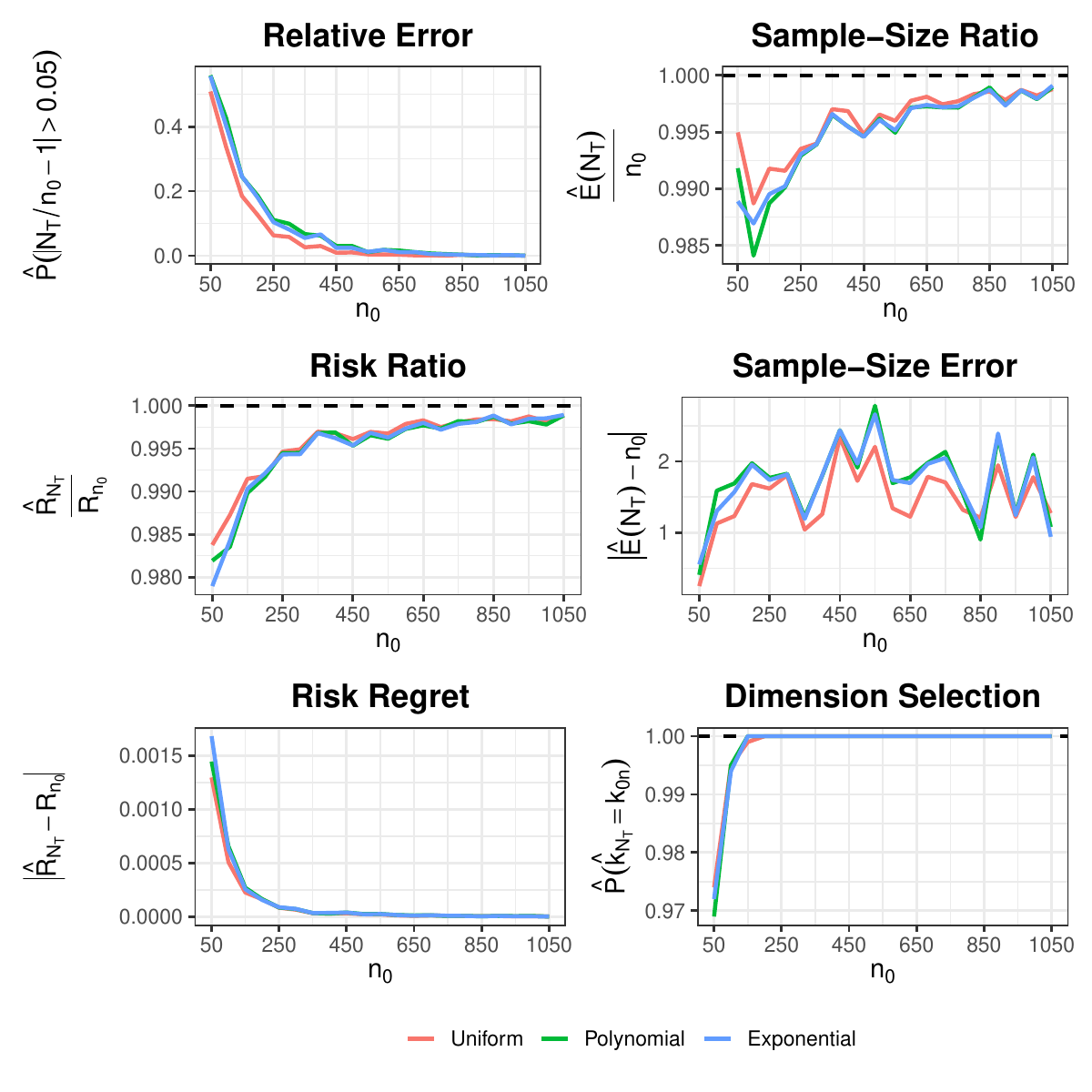}
         \caption{\(p=400\)}
        \label{fig:sim_p400}
    \end{subfigure}

    \vspace{0.5em}

    \begin{subfigure}[t]{0.48\textwidth}
        \centering
        \includegraphics[width=\textwidth]{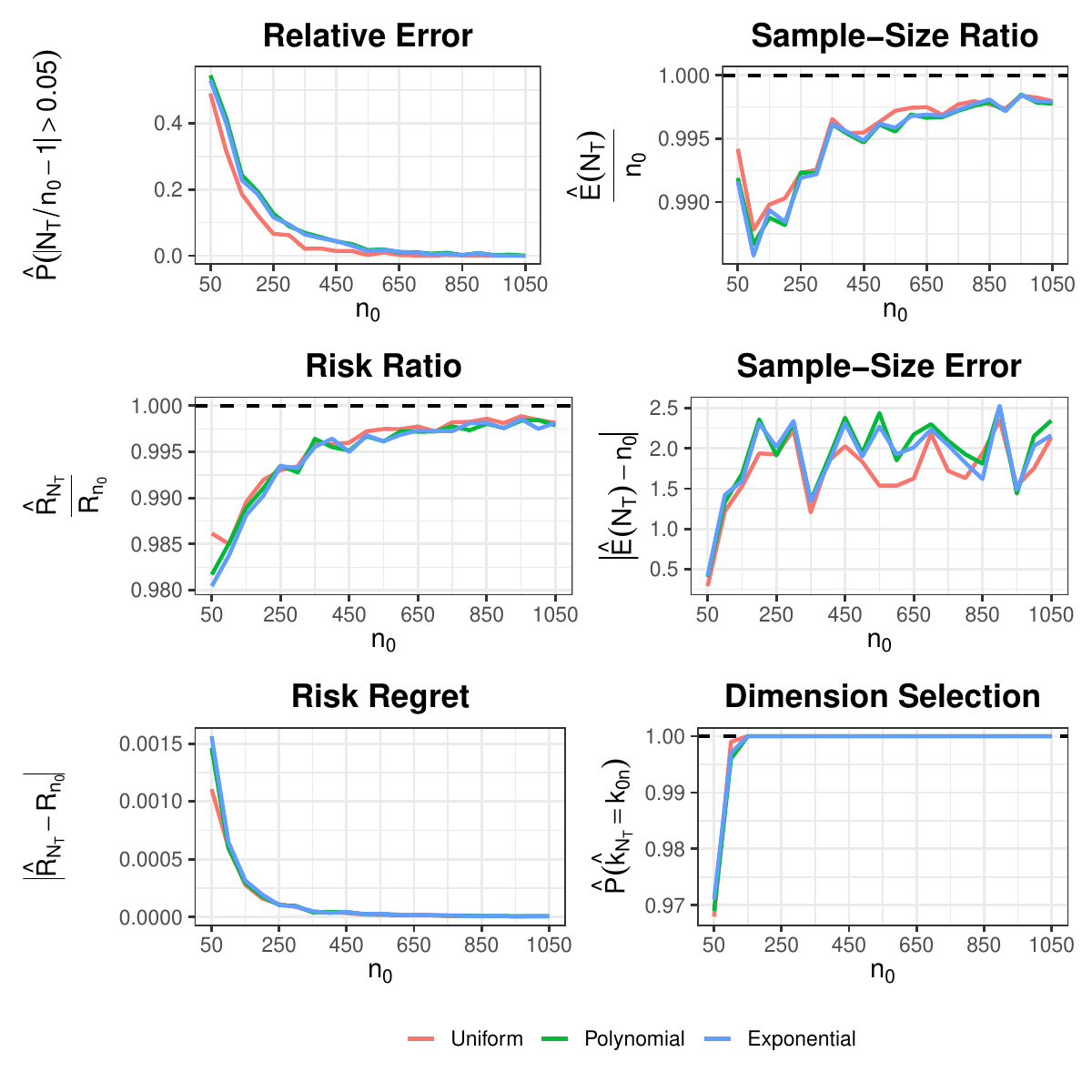}
        \caption{\(p=800\)}
        \label{fig:sim_p800}
    \end{subfigure}
    \hfill
    \begin{subfigure}[t]{0.48\textwidth}
        \centering
        \includegraphics[width=\textwidth]{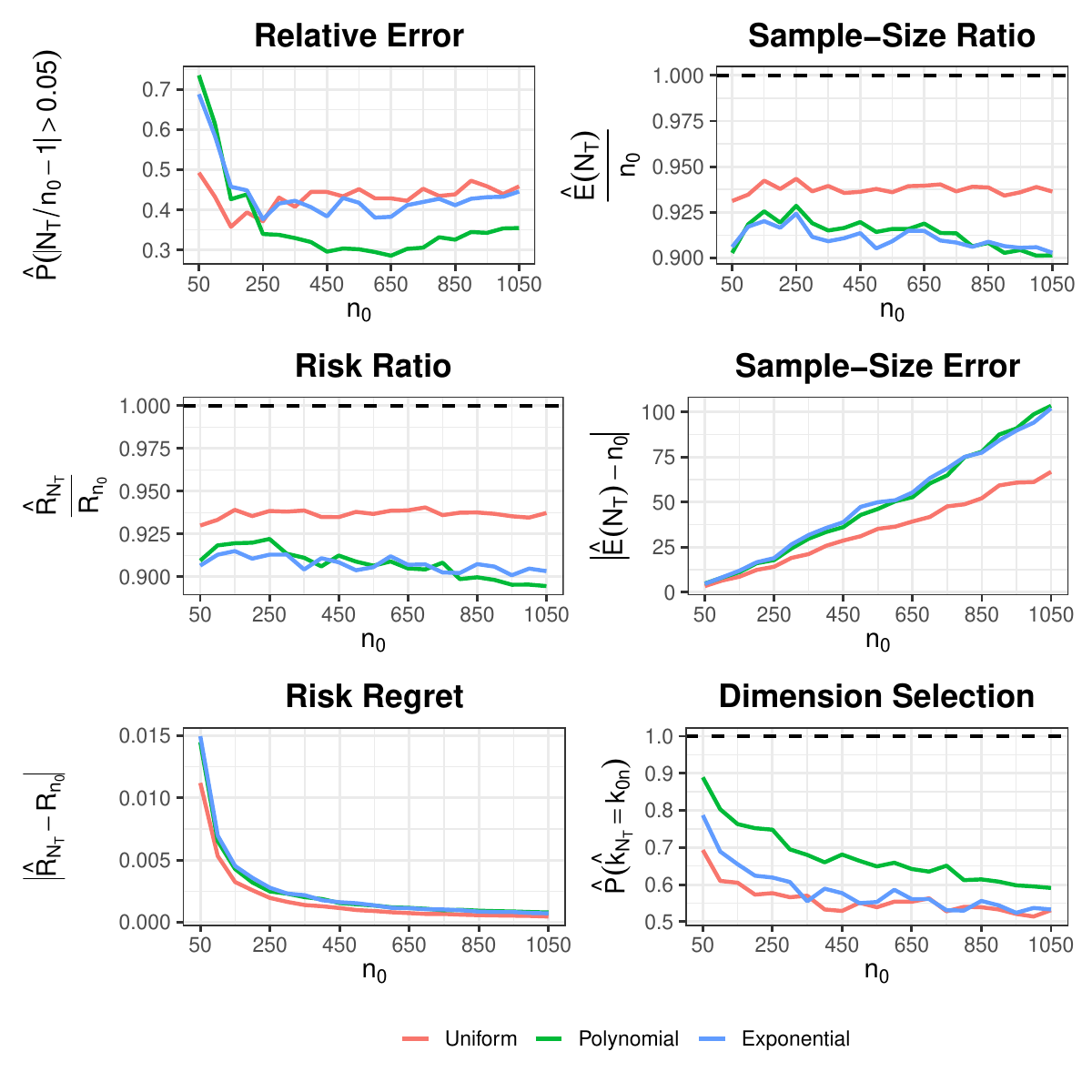}
        \caption{Gap condition violated \(p=200\)}
        \label{fig:sim_gap}
    \end{subfigure}
    \caption{Empirical performance of the proposed three-stage PCA procedure for ambient dimensions \(p=200,400,\) and \(800\), under the uniform, polynomial, and exponential tail covariance models. The bottom-right panel corresponds to a counterexample in which the gap condition is violated.}
    \label{fig:simulation_results}
\end{figure}
\vspace{0.2cm}
\noindent \textit{Importance of the Gap Condition.} To further illustrate the role of Assumption~\ref{ass:gap_k0}, we consider an additional setting in which the threshold-gap condition is violated. We take \(p=200\) and modify the three covariance models by increasing the total mass of the residual eigenvalues from \(1\) to \(3\). Specifically, the uniform-tail model is defined by \(\Lambda_{\mathrm U}=\operatorname{diag}(4,2,1,3/q,\ldots,3/q)\), while in the polynomial- and exponential-tail models we replace \(w_j^{(P)}\) and \(w_j^{(E)}\) by \(3w_j^{(P)}\) and \(3w_j^{(E)}\), respectively. Then \(\operatorname{tr}(\Sigma)=10\), so that \(\gamma_1=0.40\) and \(\gamma_2=0.60=\eta\). Hence \(k_{0n}=2\), but the threshold \(\eta\) is no longer separated from \(\gamma_2\), violating Assumption~\ref{ass:gap_k0}. Importantly, the eigengap remains \(\lambda_2-\lambda_3=1\), so this experiment isolates the effect of violating the threshold-gap condition rather than the spectral separation condition. In this setting, \(\xi_n=4\), and we therefore set \(c=4/n_0^2\) and \(m=\left\lceil(n_0^2/4)^{5/12}\right\rceil\). The corresponding population oracle risk is \(R_{n_0}=8/n_0\).

\noindent The bottom-right panel of Figure~\ref{fig:simulation_results} shows the results under this violation. In contrast to the previous settings, both first-order sample-size and risk efficiency deteriorate, the second-order sample-size error is no longer well controlled, and the probability of selecting the correct number of population principal components fails to approach one. Although the risk regret continues to decrease with \(n_0\), it does so more slowly and remains noticeably larger over the range considered. These findings illustrate the importance of the threshold-gap condition for stabilizing the estimated number of population principal components and, consequently, the sequential stopping rule.
\section{Application to TCGA Gene-Expression Data: A High-Dimensional Setting}
\label{application}

\begingroup
\captionsetup{font=footnotesize}
\setlength{\textfloatsep}{8pt plus 2pt minus 2pt}
\setlength{\intextsep}{6pt plus 2pt minus 2pt}
\setlength{\floatsep}{6pt plus 2pt minus 2pt}
\setlength{\abovecaptionskip}{4pt}
\setlength{\belowcaptionskip}{2pt}

\noindent We illustrate the proposed three-stage procedure using gene-expression data from The Cancer Genome Atlas (TCGA). This application provides a natural high-dimensional setting for our framework: each patient is represented by measurements on approximately \(p=20{,}000\) genes, whereas the number of available patients within each cancer-type cohort is substantially smaller. Our objective is to determine, for each cohort, how many patients are needed to obtain a cost-effective PCA representation while balancing the statistical loss due to compression against the cost of collecting additional observations.

\noindent \paragraph{Cohort sizes and effective ranks.}
The TCGA Pan-Cancer data contain gene-expression measurements for 32 cancer types across 9{,}701 patients \citep[e.g.,][]{Weinstein2013,GDC2026,Hoadley2018}. For a given cancer type, let \(n_{\mathrm{avail}}\) denote the total number of available patients in that cohort. The same \(p\approx20{,}000\) genes are measured for each patient, while \(n_{\mathrm{avail}}\) ranges from 36 to 1{,}097 across cohorts, so that \(p\gg n_{\mathrm{avail}}\) throughout. Expression values are standardized and log-transformed by TCGA upstream of our analysis, and we use these measurements directly. For a sample of \(j\le n_{\mathrm{avail}}\) patients, we retain the notation used throughout the paper: \(S_j\) denotes the sample covariance matrix, \(\hat k_j\) is the smallest number of principal components explaining at least a fraction \(\eta\) of the total sample variation, as defined in \eqref{eq:khatn}, and \(V_j(\hat k_j)\) denotes the corresponding unexplained variance, as defined in \eqref{V_Xi}. As shown in Figure~\ref{fig:navail}, the available sample size varies by more than 30-fold across the 32 cohorts, from \(n_{\mathrm{avail}}=36\) for cholangiocarcinoma to \(n_{\mathrm{avail}}=1{,}097\) for breast invasive carcinoma. Four of the five smallest cohorts correspond to documented rare cancers: cholangiocarcinoma, uterine carcinosarcoma, kidney chromophobe, and adrenocortical cancer \citep[e.g.,][]{Banales2020,Artioli2015,PashaiFakhri2025,MenonPrasathCorrea2025}; the fifth, diffuse large B-cell lymphoma, is not generally classified as rare \citep[e.g.,][]{Wang2023}.

\noindent To assess whether these high-dimensional cohorts exhibit the low effective-rank behavior underlying our theoretical analysis, we next compute \(\hat r_{n_{\mathrm{avail}}}=\operatorname{tr}(S_{n_{\mathrm{avail}}})/\lambda_1(S_{n_{\mathrm{avail}}})\), the sample analogue of the population effective rank defined in Section~\ref{sec:notation} and used throughout the theoretical development \citep[e.g.,][]{Vershynin2018,koltchinskii2017}. As shown in Figure~\ref{fig:reff}, \(\hat r_{n_{\mathrm{avail}}}\) lies approximately between \(3\) and \(12\) across the 32 cohorts, despite the ambient dimension being roughly \(20{,}000\). This empirical behavior is consistent with the low effective-dimensional structure underlying our theoretical framework and provides a useful diagnostic before applying the proposed sequential PCA procedure. Interestingly, four of the five cancer types with the lowest empirical effective ranks are also among the cancer types reported to have low tumor mutation burden (TMB), a quantity widely studied as a biomarker for immunotherapy response \citep[e.g.,][]{LiGaoWang2023}. Esophageal carcinoma is the only exception among these five. We regard this as an empirical observation from the TCGA data rather than a biological conclusion, and leave a systematic investigation of this relationship for future work.

\begin{figure}[htbp]
\centering
\begin{subfigure}[t]{0.48\textwidth}
\centering
\includegraphics[width=\textwidth]{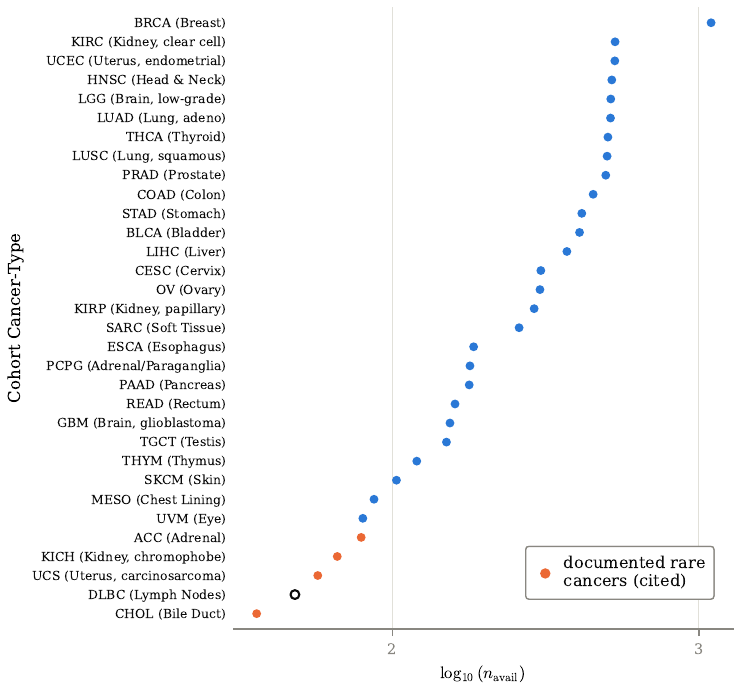}
\caption{\(\log_{10}(n_{\mathrm{avail}})\), ordered from largest to smallest. The five smallest cohorts are highlighted.}
\label{fig:navail}
\end{subfigure}
\hfill
\begin{subfigure}[t]{0.48\textwidth}
\centering
\includegraphics[width=\textwidth]{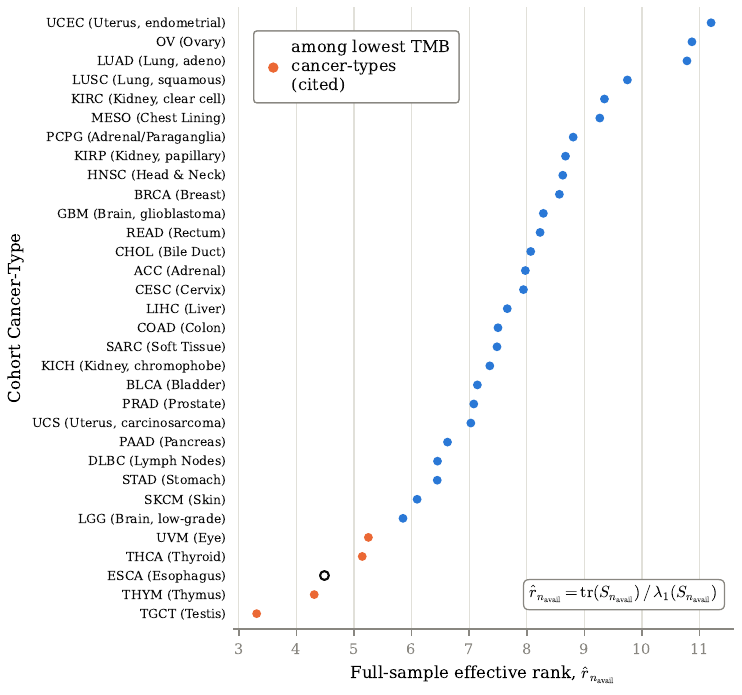}
\caption{Empirical effective rank \(\hat r_{n_{\mathrm{avail}}}\). The five lowest-rank cohorts are highlighted; the hollow marker denotes esophageal carcinoma.}
\label{fig:reff}
\end{subfigure}
\caption{Available sample size \(n_{\mathrm{avail}}\) and empirical effective rank
\(\hat r_{n_{\mathrm{avail}}}\) of the gene-expression covariance matrix across the 32 TCGA cohorts. Cohort labels use the standard TCGA study abbreviations \protect\citep[e.g.,][]{Weinstein2013,GDC2026}.}
\label{fig:cohorts}
\end{figure}

\paragraph{Sampling cost and empirical risk.}
As mentioned in Section \ref{sec:setup_oracle}, the practical motivation for sequential sampling is that genomic measurements can be costly. A recent systematic review reports genetic and genomic testing costs in oncology ranging from \$76.91 to \$11{,}431.66 \citep[e.g.,][]{Goh2026}. Accordingly, we adopt a conservative sampling cost of ($c=\$1,000$) per patient throughout, following \citet{chaban2022}, and use \(\eta=0.9\) and \(\rho=0.7\) throughout. Here \(c\) is the per-patient sampling cost, while \(A/c\) controls the relative weight assigned to unexplained PCA variation. We apply Algorithm~\ref{algo_1} separately to each cancer-type cohort to obtain its data-driven intermediate and final sample sizes. To show where the resulting stopping rule lies relative to the observable cost-compression curve, Figure~\ref{fig:costcurve_examples} considers four representative cohorts. Because TCGA is an already collected dataset with no natural prospective enrollment order, we randomly permute the patients five times to mimic different sequential enrollment orders. For each ordering and each fixed \(j\le n_{\mathrm{avail}}\), we evaluate the cost--compression loss \(L_j(\hat k_j)\) in \eqref{eq:cost_compression_loss}. The faint grey curves represent the five resulting loss trajectories, while the dashed black curve denotes their pointwise average, providing an empirical analogue of the fixed-sample risk in \eqref{fixed_risk_function}. These repetitions are used solely to assess sensitivity to patient ordering and do not form part of Algorithm~\ref{algo_1}. The colored segments show the realized three-stage sampling path \(m\rightarrow T_m\rightarrow N_T\), and the star indicates the loss obtained using the full available cohort. The full cohort is used retrospectively only to construct the dense loss trajectories; at each stage, the sequential procedure uses only the observations available up to that stage.

Figure~\ref{fig:costcurve_examples} shows that for breast invasive carcinoma, brain lower grade glioma, and bladder urothelial carcinoma, the selected \(N_T\) lies well below \(n_{\mathrm{avail}}\) and reaches the low-loss region before the full cohort is exhausted, yielding estimated dollar-risk reductions of approximately \(72\%\), \(51\%\), and \(26\%\), respectively. Uveal melanoma illustrates the opposite regime: \(n_{\mathrm{avail}}=80\), whereas the procedure recommends \(N_T\approx140\), and the loss curve is still decreasing at the available sample size, indicating that additional observations would be desirable under the specified cost-compression trade-off. Table~\ref{tab:cost_reduction} extends this analysis to all 32 cohorts using one realized sequential ordering per cohort and reports \(N_T\), the sample-size reduction \(1-N_T/n_{\mathrm{avail}}\), and the corresponding empirical dollar-risk reduction relative to using the full cohort. Hence, the exact \(N_T\) for the four cohorts in the figure need not coincide with the single-ordering values in the table. For 22 of the 32 cohorts, \(N_T\le n_{\mathrm{avail}}\), with sample-size reductions ranging from \(2.7\%\) to \(85.6\%\) and dollar-risk reductions from \(0.3\%\) to \(71.8\%\); for the remaining ten cohorts, \(N_T>n_{\mathrm{avail}}\), so no reduction is reported because the procedure instead recommends additional sampling. In particular, for breast invasive carcinoma the table gives \(N_T=158\) from \(n_{\mathrm{avail}}=1{,}097\), corresponding to an \(85.6\%\) sample-size reduction and a \(71.8\%\) estimated dollar-risk reduction.

\begin{figure}[!htbp]
\centering

\begin{subfigure}[t]{0.48\textwidth}
\centering
\includegraphics[width=\textwidth]
{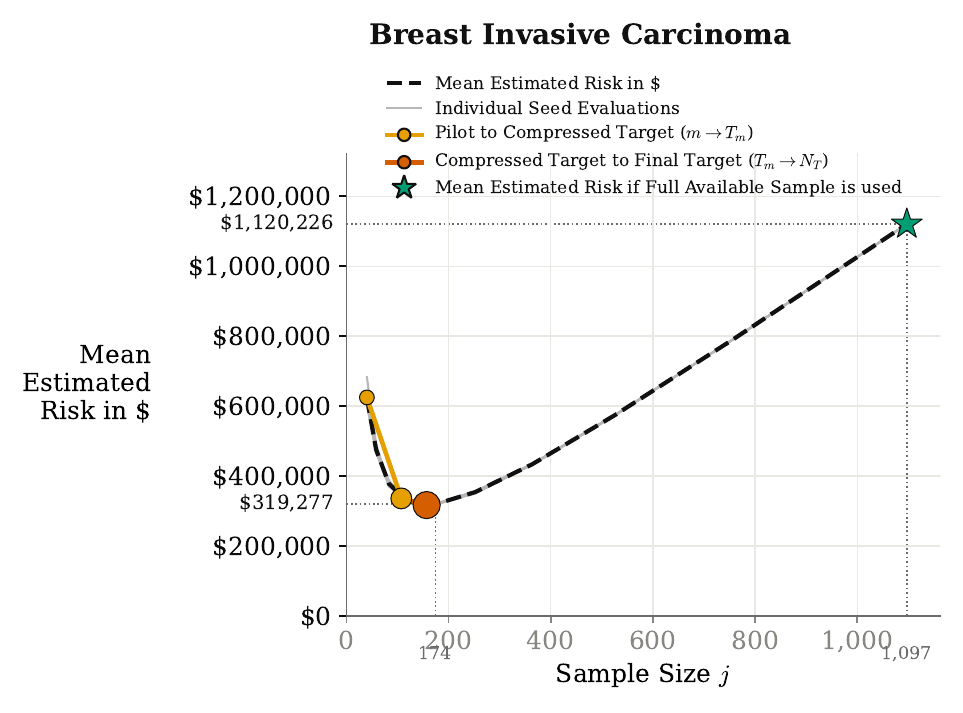}
\caption{Breast invasive carcinoma
(\(n_{\mathrm{avail}}=1{,}097\)).}
\label{fig:costcurve_brca}
\end{subfigure}
\hfill
\begin{subfigure}[t]{0.48\textwidth}
\centering
\includegraphics[width=\textwidth]
{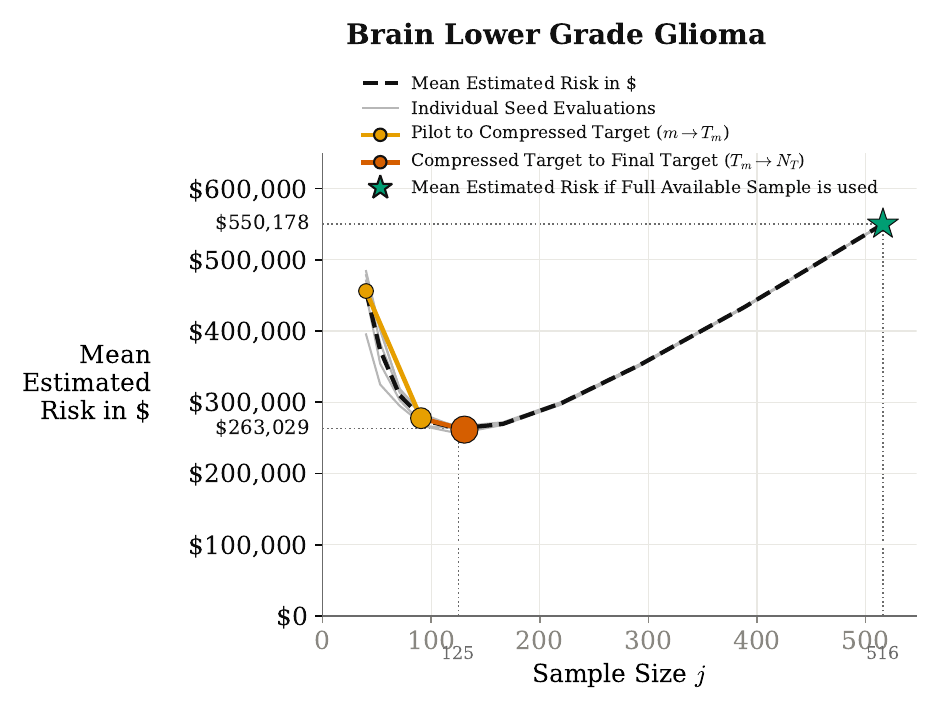}
\caption{Brain lower grade glioma
(\(n_{\mathrm{avail}}=516\)).}
\label{fig:costcurve_lgg}
\end{subfigure}

\vspace{0.5em}

\begin{subfigure}[t]{0.48\textwidth}
\centering
\includegraphics[width=\textwidth]
{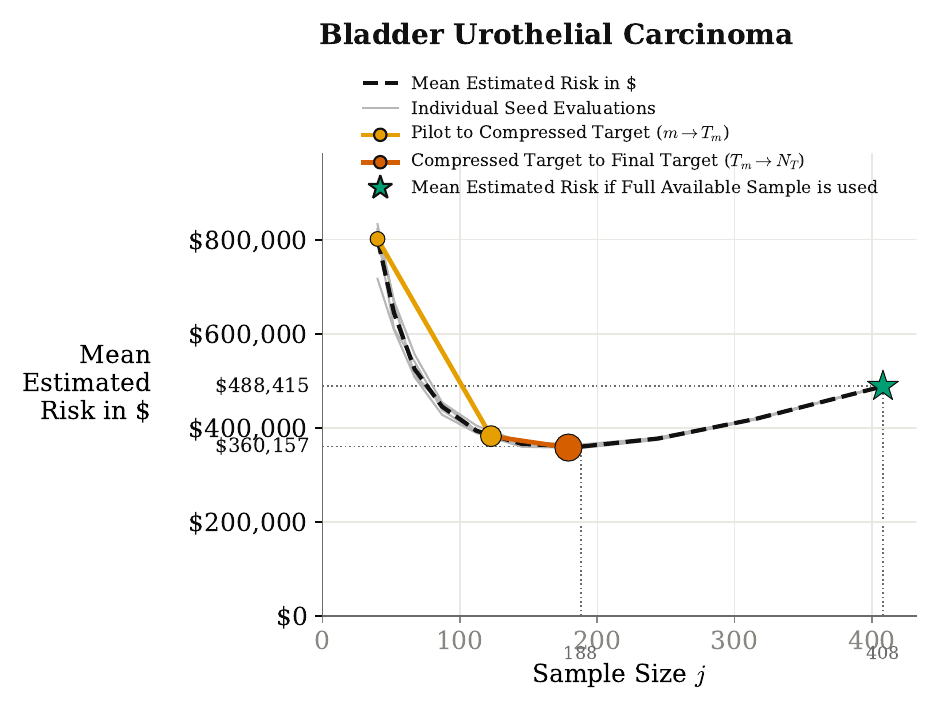}
\caption{Bladder urothelial carcinoma
(\(n_{\mathrm{avail}}=408\)).}
\label{fig:costcurve_blca}
\end{subfigure}
\hfill
\begin{subfigure}[t]{0.48\textwidth}
\centering
\includegraphics[width=\textwidth]
{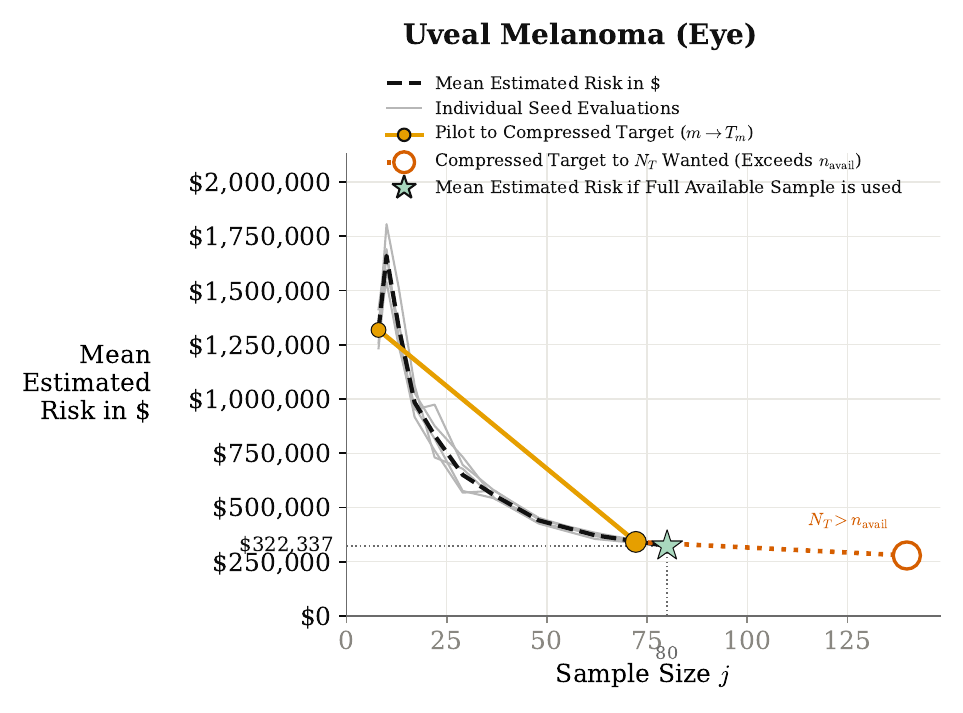}
\caption{Uveal melanoma
(\(n_{\mathrm{avail}}=80\)).}
\label{fig:costcurve_uvm}
\end{subfigure}

\caption{Empirical cost-compression curves for four representative TCGA cohorts at
\(\eta=0.9\), \(A/c=10\), and \(c=\$1{,}000\) per patient. Faint curves correspond to individual sequential orderings and the dashed curve gives their mean. The colored segments show the realized three-stage path \(m\rightarrow T_m\rightarrow N_T\), while the star denotes the empirical loss from using the full available cohort. For uveal melanoma, the recommended final target satisfies \(N_T>n_{\mathrm{avail}}\).}
\label{fig:costcurve_examples}
\end{figure}

{\small
\setlength{\tabcolsep}{4.5pt}
\renewcommand{\arraystretch}{0.92}

\begin{longtable}{@{}p{0.31\textwidth}lrrrr@{}}
\caption{Sample-size and estimated dollar-risk reductions across TCGA cohorts at
\(\eta=0.9\), \(A/c=10\), \(\rho=0.7\) and \(c=\$1{,}000\) per patient.
\(N_T\) denotes the final target selected by Algorithm~\ref{algo_1}.
Reductions are reported only when \(N_T\le n_{\mathrm{avail}}\).
Within each panel, cohorts are ordered by \(n_{\mathrm{avail}}\).
\(^{*}\)Documented rare cancer
\protect\citep[e.g.,][]{Banales2020,Artioli2015,PashaiFakhri2025,MenonPrasathCorrea2025};
\(^{\dagger}\)documented low tumor mutation burden
\protect\citep[e.g.,][]{LiGaoWang2023}.}
\label{tab:cost_reduction}\\

\toprule
Cohort & Abbrev. & \(n_{\mathrm{avail}}\) & \(N_T\) &
\multicolumn{1}{c}{Sample reduction (\%)} &
\multicolumn{1}{c}{Risk reduction (\%)} \\
\midrule
\endfirsthead

\multicolumn{6}{c}{\tablename~\thetable\ continued} \\[2pt]
\toprule
Cohort & Abbrev. & \(n_{\mathrm{avail}}\) & \(N_T\) &
\multicolumn{1}{c}{Sample reduction (\%)} &
\multicolumn{1}{c}{Risk reduction (\%)} \\
\midrule
\endhead

\midrule
\multicolumn{6}{r}{\textit{continued on next page}}\\
\endfoot

\bottomrule
\endlastfoot

\multicolumn{6}{@{}l}{\textit{Panel A: Selected target does not exceed the available cohort}}\\
\addlinespace[2pt]

glioblastoma multiforme
& GBM & 154 & 137 & 11.0 & 0.3 \\

rectum adenocarcinoma
& READ & 160 & 130 & 18.8 & 1.6 \\

pancreatic adenocarcinoma
& PAAD & 178 & 147 & 17.4 & 2.5 \\

pheochromocytoma \& paraganglioma
& PCPG & 179 & 148 & 17.3 & 2.3 \\

esophageal carcinoma
& ESCA & 184 & 179 & 2.7 & 0.4 \\

sarcoma
& SARC & 259 & 186 & 28.2 & 5.4 \\

kidney papillary cell carcinoma
& KIRP & 290 & 153 & 47.2 & 16.4 \\

ovarian serous cystadenocarcinoma
& OV & 303 & 148 & 51.2 & 20.0 \\

cervical \& endocervical cancer
& CESC & 305 & 175 & 42.6 & 14.3 \\

liver hepatocellular carcinoma
& LIHC & 371 & 173 & 53.4 & 23.1 \\

bladder urothelial carcinoma
& BLCA & 408 & 182 & 55.4 & 26.0 \\

stomach adenocarcinoma
& STAD & 415 & 145 & 65.1 & 37.6 \\

colon adenocarcinoma
& COAD & 452 & 138 & 69.5 & 43.8 \\

prostate adenocarcinoma
& PRAD & 497 & 125 & 74.8 & 53.1 \\

lung squamous cell carcinoma
& LUSC & 502 & 166 & 66.9 & 41.1 \\

thyroid carcinoma$^{\dagger}$
& THCA & 505 & 126 & 75.0 & 53.4 \\

lung adenocarcinoma
& LUAD & 515 & 155 & 69.9 & 44.5 \\

brain lower grade glioma
& LGG & 516 & 131 & 74.6 & 51.4 \\

head \& neck squamous cell carcinoma
& HNSC & 520 & 161 & 69.0 & 43.9 \\

uterine corpus endometrioid carcinoma
& UCEC & 532 & 159 & 70.1 & 44.2 \\

kidney clear cell carcinoma
& KIRC & 533 & 142 & 73.4 & 49.6 \\

breast invasive carcinoma
& BRCA & 1097 & 158 & 85.6 & 71.8 \\

\addlinespace[4pt]
\midrule
\multicolumn{6}{@{}l}{\textit{Panel B: Selected target exceeds the available cohort}}\\
\addlinespace[2pt]

cholangiocarcinoma$^{*}$
& CHOL & 36 & 119 & -- & -- \\

diffuse large B-cell lymphoma
& DLBC & 48 & 99 & -- & -- \\

uterine carcinosarcoma$^{*}$
& UCS & 57 & 117 & -- & -- \\

kidney chromophobe$^{*}$
& KICH & 66 & 77 & -- & -- \\

adrenocortical cancer$^{*}$
& ACC & 79 & 101 & -- & -- \\

uveal melanoma$^{\dagger}$
& UVM & 80 & 140 & -- & -- \\

mesothelioma
& MESO & 87 & 151 & -- & -- \\

skin cutaneous melanoma
& SKCM & 103 & 113 & -- & -- \\

thymoma$^{\dagger}$
& THYM & 120 & 158 & -- & -- \\

testicular germ cell tumor$^{\dagger}$
& TGCT & 150 & 166 & -- & -- \\

\end{longtable}
}

\noindent Overall, the TCGA analysis illustrates both possible outcomes of the proposed procedure: when \(N_T\le n_{\mathrm{avail}}\), it yields a smaller cost-effective sample size while preserving the desired PCA compression, whereas \(N_T>n_{\mathrm{avail}}\) suggests that additional observations would be useful. Thus, the application complements Sections~\ref{sec:three_stage} and~\ref{sec:asym_efficiency} by demonstrating the practical value of the proposed rule for high-dimensional study design. When \(n_{\mathrm{avail}}\gg N_T\), the selected sample could also serve as training data for prospective predictive modeling (e.g., classification or regression models), with the remaining observations reserved for testing; this idea extends naturally beyond the TCGA setting.
\endgroup

\section{Concluding Remarks and Future Directions}
\label{sec:conclusion}

\noindent This paper developed a three-stage adaptive PCA procedure for sequentially arriving high-dimensional data under a cost-compression framework. We formulated the corresponding oracle problem and proposed a three-stage sampling rule that determines
the final sample size using data collected across the three stages. We showed that the rule collects only a finite sample size $N_T$ almost surely. Additionally, we established first-order efficiency in both sample size and risk, together with the second-order guarantees, that is, \(\left|\mathbb E_0[N_T]-n_0(c)\right|=O(1)\) and \(\left|R_{N_T}(\hat k_{N_T})-R_{n_0(c)}\right|=O(c)\) (Section~\ref{sec:asym_efficiency}). The comparison with the two-stage rule further demonstrates the benefit of the final recalibration step. The simulation studies in Section~5 support these theoretical findings in finite samples, while the TCGA gene-expression analysis in Section~\ref{application} illustrates how the proposed procedure can be used as a practical study-design tool in a \(p\gg n\) setting, both to identify cost-effective early stopping and to indicate when additional observations are desirable.

Several directions remain open for future work. Extending the theory beyond i.i.d.\ sub-Gaussian observations to dependent or heavy-tailed data would broaden the scope of the method. It is also natural to investigate whether additional recalibration stages can further improve efficiency and whether an optimal number of stages can be characterized. Finally, extending the proposed cost-compression framework beyond PCA to other sequential high-dimensional dimension-reduction and estimation problems is a promising direction for future research.

\section*{Funding Statement}

\noindent Sarkar’s work was supported by the National Science Foundation (NSF-DMS-2506060).

\section*{Declaration of Competing Interest}

\noindent The authors declare that they have no known competing financial interests or personal relationships that could have appeared to influence the work reported in this paper.

\section*{Declaration of generative AI and AI-assisted technologies in the manuscript preparation process}

\noindent During the preparation of this work, the author(s) used ChatGPT (OpenAI) and Claude (Anthropic) for checking, improving writing flow and clarity, and assistance with coding. After using these tools, the author(s) carefully reviewed, edited, and verified the accuracy, comprehensiveness, and impartiality of the AI-assisted output as needed, and take full responsibility for the content of the published article.

\section*{Acknowledgements}

\noindent We are grateful to Swarnali Banerjee, Dept. of Math and Stat, Data Science, Loyola University Chicago, for valuable discussions on the cost-compression risk framework and for her helpful comments and feedback on this work.

\section*{Supplementary Material}

\noindent Additional technical details and supporting results are provided in the
supplementary material accompanying this article \citep[e.g.,][]{SarkarEtAlSupplement}.
The supplement includes background on sub-Gaussian random variables and
vectors, together with additional consistency results for the PCA dimension,
principal subspace, and residual variation at the random final sample size
and at the corresponding oracle sample size.

\appendix

\section{Proofs of Main Results Stated in the Article}\label{app_A}

\subsection{Proof of Lemma~\ref{lem:oracle_optimality}}

\noindent By Assumption~\ref{ass:non_incr_tail_mass} it follows \(R_n^\star\ge cn\to\infty\) as \(n\to\infty\). Also, by monotonicity, \(\xi_n\le\xi_1\), so \(R_n^\star\le A\xi_1/n+cn\). Choose \(N\) such that \(cN>R_1^\star\). Then \(R_n^\star>R_1^\star\) for all \(n\ge N\), and hence $\inf_{n\in\mathbb N}R_n^\star=\min_{1\le n\le N}R_n^\star.$
Thus the set of minimizers is nonempty, and defining \(n_0(c)\) as its smallest element makes it well-defined and unique.

\noindent It remains to show \(n_0(c)\to\infty\) as \(c\downarrow0\). Fix \(M\in\mathbb N\) and choose \(n>M\) such that \(A\xi_1/n\le A\xi_M/(2M)\). Next choose \(c_M>0\) so that \(c_Mn\le A\xi_M/(2M)\). Then, for \(0<c\le c_M\) and any \(k\le M\),
\[
R_n^\star
\le \frac{A\xi_1}{n}+cn
\le \frac{A\xi_M}{M}
\le \frac{A\xi_M}{k}
< \frac{A\xi_k}{k}+ck
=R_k^\star.
\]
Hence \(n_0(c)>M\) for all \(0<c\le c_M\). Since \(M\) is arbitrary, \(n_0(c)\to\infty\) as \(c\downarrow0\).

\subsection{Proof of Lemma~\ref{lem:asymp_equi}}
\noindent Write \(n_0=n_0(c)\). By Lemma~\ref{lem:oracle_optimality}, \(n_0\to\infty\) as \(c\downarrow0\). Optimality, \(R_{n_0}^\star\le R_{n_0\pm1}^\star\), and Assumption~\ref{ass:non_incr_tail_mass} give
\[
c\ge \frac{A\xi_{n_0}}{n_0(n_0+1)},\qquad
c\le \frac{A\xi_{n_0}}{n_0(n_0-1)}
+\frac{A(\xi_{n_0-1}-\xi_{n_0})}{n_0-1}.
\]
By Assumption~\ref{ass:oracle_flatness}, \(\xi_{n_0-1}=\xi_{n_0}\{1+o(1/n_0)\}\), so the second term in the upper bound is \(\xi_{n_0}o(n_0^{-2})\). Hence \(c=A\xi_{n_0}/n_0^2\{1+o(1)\}\).

\noindent For the second claim, recall \(n_T(c)=\min\{n:A\xi_n/n^2\le c\}\). The preceding lower bound and monotonicity imply
\[
\frac{A\xi_{n_0+1}}{(n_0+1)^2}
\le \frac{A\xi_{n_0}}{(n_0+1)^2}\le c,
\]
so \(n_T(c)\le n_0+1\). Conversely, for any fixed \(\epsilon\in(0,1)\) and every integer \(n\le(1-\epsilon)n_0\),
\[
\frac{A\xi_n}{n^2}
\ge \frac{A\xi_{n_0}}{(1-\epsilon)^2n_0^2}
=\frac{c\{1+o(1)\}}{(1-\epsilon)^2}>c
\]
for sufficiently small \(c\). Thus \(n_T(c)>(1-\epsilon)n_0\). Since \(\epsilon>0\) is arbitrary and \(n_0\to\infty\), \(n_T(c)/n_0(c)\to1\), completing the proof.

\subsection{Proof of Lemma \ref{lem:sampleext}}
\noindent Since \(X_i\in\mathbb R^{p_n}\) almost surely, \(\|X_i\|_2<\infty\) almost surely, and hence \(tr(S_n)=n^{-1}\sum_{i=1}^n\|X_i\|_2^2<\infty\) almost surely for every fixed \(n\). Therefore, $P\!\left(\bigcap_{n\ge1}\{tr(S_n)<\infty\}\right)=1.$
On this probability-one event, the pilot-stage quantity \(V_m\), being a finite-valued functional of \(S_m\), is finite, and hence the intermediate sample size \(T\), obtained from \(V_m\) through finitely many arithmetic operations and ceiling/max operators, is also finite. Since \(tr(S_T)<\infty\), the corresponding quantity \(V_T\) is finite as well. The final sample size \(N_T\), being obtained from \(V_m\) and \(V_T\) by the same finite operations, is therefore finite. Hence \(P(N_T<\infty)=1\).

\subsection{Proof of Theorem \ref{th:th_almost_sure_cons}}

\noindent Fix \(\epsilon\in(0,1)\). From definition of $N_T$ it follows
{\small
\[
\mathbb P_0\!\left(\left|\frac{N_T}{n_0(c)}-1\right|>\epsilon\right)
\le
\mathbb P_0\!\left(\sqrt{\frac{A}{c}V_T(\hat k_T)}<(1-\epsilon)n_0(c)\right)
+\mathbb P_0\!\left(\frac{T}{n_0(c)}>1+\epsilon\right)
+\mathbb P_0\!\left(\sqrt{\frac{A}{c}V_T(\hat k_T)}>(1+\epsilon)n_0(c)-1\right).
\]
}
Since \(\rho\in(0,1)\), Lemma~\ref{lem:T_cont} gives
$\mathbb P_0\!\left(T/n_0(c)>1+\epsilon\right)
\precsim e^{-K_2m}$. Next choose \(\delta>0\) such that \(\rho(1+\delta)<1\). On the event \(|T/\{\rho n_0(c)\}-1|<\delta\), we have \(T\ge m\) and, by monotonicity together with Assumption~\ref{as:A2},
\[
1\le\frac{\xi_T}{\xi_{n_0(c)}}\le\frac{\xi_m}{\xi_{n_0(c)}}\to1.
\]
Hence, using Lemma~\ref{lem:asymp_equi}, \(n_0(c)^2c/A=\xi_T\{1+o(1)\}\) uniformly on this event. Therefore Lemma~\ref{lem:rel_cont_v}, conditioning on \(T=t\), yields
\begin{equation}\label{N_T_cont4}
\mathbb P_0\!\left(V_T(\hat k_T)<(1-\epsilon)^2\frac{n_0(c)^2c}{A}\right)
\precsim e^{-K_2m},
\end{equation}
where we also use \(T\ge m\) and Lemma~\ref{lem:T_cont} to control the complement of \(|T/\{\rho n_0(c)\}-1|<\delta\).

\noindent Similarly, since \(\{(1+\epsilon)n_0(c)-1\}^2=(1+\epsilon)^2n_0(c)^2\{1+o(1)\}\), the same argument gives
\begin{equation}\label{N_T_cont5}
\mathbb P_0\!\left(V_T(\hat k_T)>(1+\epsilon)^2\frac{n_0(c)^2c}{A}\right)
\precsim e^{-K_2m}.
\end{equation}
Combining \eqref{N_T_cont4}, and~\eqref{N_T_cont5} gives the desired probability bound. The almost sure convergence is then an immediate consequence of the Borel--Cantelli lemma.
\subsection{Proof of Theorem~\ref{th:mean_cons}}

\noindent It is enough to show
\[
\mathbb E_0\left[\left|\frac{N_T}{n_0(c)}-1\right|\right]\to0
\qquad\text{as }c\downarrow0.
\]
From definition of $N_T$ together with \(|\lceil x\rceil-x|\le1\) and \(|\lceil x\rceil-\lceil y\rceil|\le|x-y|+1\), we obtain
{\small
\[
\mathbb E_0\left|\frac{N_T}{n_0(c)}-1\right|
\le
\mathbb E_0\left|\frac{1}{n_0(c)}\sqrt{\frac{A}{c}V_T(\hat k_T)}-1\right|
+\frac{m}{n_0(c)}
+\frac{1}{n_0(c)}
\mathbb E_0\left|\sqrt{\frac{A}{c}V_m(\hat k_m)}-\sqrt{\frac{A}{c}V_T(\hat k_T)}\right|
+\frac{2}{n_0(c)}.
\]
}
\noindent By Lemma~\ref{lem:asymp_equi}, \(A/\{cn_0(c)^2\}=\xi_{n_0(c)}^{-1}\{1+o(1)\}\). Moreover, \(m=o\{n_0(c)\}\) by Assumption~\ref{as:A1} and \(n_0(c)\to\infty\). Hence it suffices to show
\begin{equation}\label{exp_sample_size1}
\mathbb E_0\left|\frac{V_m(\hat k_m)-\xi_{n_0(c)}}{\xi_{n_0(c)}}\right|\to0,
\qquad
\mathbb E_0\left|\frac{V_T(\hat k_T)-\xi_{n_0(c)}}{\xi_{n_0(c)}}\right|\to0.
\end{equation}

\noindent For the first term, Assumption~\ref{as:A2} and Lemma~\ref{lem:rel_cont_v} yield
\begin{align}
\mathbb E_0\left|\frac{V_m(\hat k_m)-\xi_{n_0(c)}}{\xi_{n_0(c)}}\right|
\le
\frac{\xi_m}{\xi_{n_0(c)}}
\mathbb E_0\left|\frac{V_m(\hat k_m)-\xi_m}{\xi_m}\right|
+\left|\frac{\xi_m}{\xi_{n_0(c)}}-1\right| \precsim
\frac{\xi_m}{\xi_{n_0(c)}}m^{-1/2}
+\left|\frac{\xi_m}{\xi_{n_0(c)}}-1\right|
\to0.
\label{exp_sample_size2}
\end{align}

\noindent For the random index \(T\), conditioning on \(T\) and using Lemma~\ref{lem:rel_cont_v} gives
\begin{align}
\mathbb E_0\left|\frac{V_T(\hat k_T)-\xi_{n_0(c)}}{\xi_{n_0(c)}}\right|
\le
\mathbb E_0\left[
\frac{\xi_T}{\xi_{n_0(c)}}
\left|\frac{V_T(\hat k_T)-\xi_T}{\xi_T}\right|
\right]
+\mathbb E_0\left|\frac{\xi_T}{\xi_{n_0(c)}}-1\right| \precsim
\sup_{t\ge m}\frac{\xi_t}{\xi_{n_0(c)}}\,m^{-1/2}
+\mathbb E_0\left|\frac{\xi_T}{\xi_{n_0(c)}}-1\right|.
\label{exp_sample_size3}
\end{align}
The first term vanishes by Assumptions~\ref{as:A2} and~\ref{as:B1}. To control the second, choose \(\epsilon>0\) such that $\rho(1-\epsilon)>1-\alpha$ and $\rho(1+\epsilon)<1+\alpha,$
where \(\alpha\) is from Assumption~\ref{as:B1}. Then
\begin{align}
\mathbb E_0\left|\frac{\xi_T}{\xi_{n_0(c)}}-1\right|
\le
\sup_{\left|t/\{\rho n_0(c)\}-1\right|\le\epsilon}
\left|\frac{\xi_t}{\xi_{n_0(c)}}-1\right| +
\left(\sup_{t\ge1}\frac{\xi_t}{\xi_{n_0(c)}}+1\right)
\mathbb P_0\left(\left|\frac{T}{\rho n_0(c)}-1\right|>\epsilon\right).
\label{exp_sample_size4}
\end{align}
On the event \(|T/\{\rho n_0(c)\}-1|\le\epsilon\), the index \(T\) lies in \([(1-\alpha)n_0(c),(1+\alpha)n_0(c)]\). Thus the first term in \eqref{exp_sample_size4} tends to zero by Assumptions~\ref{ass:non_incr_tail_mass} and~\ref{as:B1}, while the second tends to zero by Assumption~\ref{as:B1} and Lemma~\ref{lem:T_cont}. Therefore the second term in \eqref{exp_sample_size1} also converges to zero. Consequently, \(\mathbb E_0|N_T/n_0(c)-1|\to0\), which implies \(\mathbb E_0[N_T]/n_0(c)\to1\) and completes the proof.

\subsection{Proof of Theorem~\ref{th:risk_cons}}

\noindent Write \(n_0=n_0(c)\) and \(\xi_0=\xi_{n_0(c)}\). By the definition of the risk,
\[
R_{N_T}(\hat k_{N_T})
=A\mathbb E_0\!\left[\frac{V_{N_T}(\hat k_{N_T})}{N_T}\right]
+c\mathbb E_0[N_T],
\qquad
R_{n_0(c)}=\frac{A\xi_0}{n_0}+cn_0.
\]
Hence
\begin{equation}
\label{eq:risk_cons_1}
\left|\frac{R_{N_T}(\hat k_{N_T})}{R_{n_0(c)}}-1\right|
\le
\frac{n_0}{\xi_0}
\mathbb E_0\left|
\frac{V_{N_T}(\hat k_{N_T})}{N_T}-\frac{\xi_0}{n_0}
\right|
+\mathbb E_0\left|\frac{N_T}{n_0}-1\right|.
\end{equation}
The second term vanishes by Theorem~\ref{th:mean_cons}. Adding and subtracting \(\xi_{N_T}/N_T\) and \(\xi_0/N_T\), the first term is bounded by \(I_{1c}+I_{2c}+I_{3c}\), where
\[
I_{1c}:=\mathbb E_0\!\left[
\left|\frac{V_{N_T}(\hat k_{N_T})-\xi_{N_T}}{\xi_{N_T}}\right|
\frac{\xi_{N_T}}{\xi_0}\frac{n_0}{N_T}\right],\quad
I_{2c}:=\mathbb E_0\!\left[
\frac{n_0}{N_T}\left|\frac{\xi_{N_T}}{\xi_0}-1\right|\right],\quad
I_{3c}:=\mathbb E_0\left|\frac{n_0}{N_T}-1\right|.
\]

We first record a bound used repeatedly below. For any fixed \(\epsilon\in(0,1)\), since \(N_T\ge m\),
\begin{equation}
\label{eq:risk_cons_N_inverse_bound}
\mathbb E_0\!\left[\frac{n_0^2}{N_T^2}\right]
\le \frac{1}{(1-\epsilon)^2}
+\frac{n_0^2}{m^2}
\mathbb P_0\!\left(\left|\frac{N_T}{n_0}-1\right|\ge\epsilon\right)
=O(1).
\end{equation}
Indeed, Theorem~\ref{th:th_almost_sure_cons} gives the probability bound \(O(e^{-K_1m})\), and with \(m=(A/c)^{1/(2\delta)}\), the exponential term dominates the corresponding polynomial factor as \(c\downarrow0\).

For \(I_{1c}\), Cauchy--Schwarz, monotonicity of \(\xi_t\), and \eqref{eq:risk_cons_N_inverse_bound} give
\[
I_{1c}
\le
\left\{\mathbb E_0\left[
\left(\frac{V_{N_T}(\hat k_{N_T})-\xi_{N_T}}{\xi_{N_T}}\right)^2
\right]\right\}^{1/2}
\frac{\xi_m}{\xi_0}
\left\{\mathbb E_0\left[\frac{n_0^2}{N_T^2}\right]\right\}^{1/2}.
\]
By Lemma~\ref{lem:rel_cont_v} and \(N_T\ge m\), the first factor is \(O(m^{-1/2})\), while \(\xi_m/\xi_0=O(1)\) by Assumption~\ref{as:A2}. Thus \(I_{1c}\to0\).

Next fix \(\epsilon\in(0,\alpha)\) and let \(\mathcal E_c=\{|N_T/n_0-1|<\epsilon\}\). Then
\begin{align}
I_{2c}
\le
\sup_{|t/n_0-1|<\epsilon}
\left|\frac{\xi_t}{\xi_0}-1\right|
\mathbb E_0\left[\frac{n_0}{N_T}\right] +
\left(\sup_{t\ge1}\frac{\xi_t}{\xi_0}+1\right)
\mathbb E_0\left[\frac{n_0}{N_T}\mathbf 1_{\mathcal E_c^c}\right].
\label{eq:risk_cons_I2_bound}
\end{align}
The first expectation is \(O(1)\) by \eqref{eq:risk_cons_N_inverse_bound}, while
\[
\mathbb E_0\left[\frac{n_0}{N_T}\mathbf 1_{\mathcal E_c^c}\right]
\le \frac{n_0}{m}\mathbb P_0(\mathcal E_c^c)
\precsim \frac{n_0}{m}e^{-K_1m}\to0.
\]
Assumptions~\ref{ass:non_incr_tail_mass} and~\ref{as:B1} also give
\(\sup_{|t/n_0-1|<\epsilon}|\xi_t/\xi_0-1|\to0\) and
\(\sup_{t\ge1}\xi_t/\xi_0=O(1)\). Hence \(I_{2c}\to0\).

Finally,
\[
I_{3c}
=\mathbb E_0\left|\frac{n_0}{N_T}-1\right|
\le
\frac{1}{1-\epsilon}\mathbb E_0\left|\frac{N_T}{n_0}-1\right|
+\left(\frac{n_0}{m}+1\right)\mathbb P_0(\mathcal E_c^c).
\]
The first term vanishes by Theorem~\ref{th:mean_cons}, while the second is \(O\{(n_0/m+1)e^{-K_1m}\}=o(1)\) by Theorem~\ref{th:th_almost_sure_cons}. Therefore \(I_{3c}\to0\). 
\noindent Thus \(I_{1c}+I_{2c}+I_{3c}\to0\), and \eqref{eq:risk_cons_1} yields
\[
\frac{R_{N_T}(\hat k_{N_T})}{R_{n_0(c)}}\to1,
\]
which completes the proof.

\subsection{Proof of Theorem~\ref{th:mean_diff}}

Since \(N_T=\max\{T_m,X_T\}=X_T+(T_m-X_T)_+\), where $X_T:=\left\lceil\sqrt{(A/c)V_T(\hat k_T)}\right\rceil,$ it follows,
\begin{equation}
\label{eq:mean_diff_1}
\bigl|\mathbb E_0[N_T]-n_0(c)\bigr|
\le
\bigl|\mathbb E_0[X_T]-n_0(c)\bigr|
+\mathbb E_0[(T_m-X_T)_+].
\end{equation}
We first show that the second term is \(o(1)\). Let
\(A_T=\{|T_m/\{\rho n_0(c)\}-1|\le\epsilon\}\), where \(\epsilon>0\) is chosen so that \(\rho(1+\epsilon)<1\). Then
\[
\mathbb E_0[(T_m-X_T)_+]
\le
\mathbb E_0[T_m\mathbf 1_{\{X_T<T_m\}}\mathbf 1_{A_T}]
+\mathbb E_0[T_m\mathbf 1_{A_T^c}].
\]
On \(A_T\), \(T_m\le(1-\zeta)n_0(c)\) for some \(\zeta>0\); hence, by the same argument as in Lemma~\ref{lem:T_cont},
\[
\mathbb E_0[T_m\mathbf 1_{\{X_T<T_m\}}\mathbf 1_{A_T}]
\precsim n_0(c)e^{-K_2m}=o(1).
\]
For the second term, the definition of \(T_m\) and Lemma~\ref{lem:T_cont} give
\[
\mathbb E_0[T_m\mathbf 1_{A_T^c}]
\precsim
\{m+n_0(c)+1\}e^{-K_2m}
+\mathbb E_0\!\left[
\left|\sqrt{\frac{A}{c}V_m(\hat k_m)}-n_0(c)\right|
\mathbf 1_{A_T^c}\right].
\]
Writing \(b_c=A\xi_{n_0(c)}/\{cn_0(c)^2\}\), the last expectation is bounded by
\[
n_0(c)\mathbb E_0\!\left[
\left|\frac{V_m(\hat k_m)-\xi_{n_0(c)}}{\xi_{n_0(c)}}\right|
\mathbf 1_{A_T^c}\right]
+n_0(c)|\sqrt{b_c}-1|e^{-\kappa_1m}.
\]
The second term is \(o(1)\), while Cauchy--Schwarz, Lemma~\ref{lem:rel_cont_v}, Assumption~\ref{as:A2}, and Lemma~\ref{lem:T_cont} yield
\[
n_0(c)\mathbb E_0\!\left[
\left|\frac{V_m(\hat k_m)-\xi_{n_0(c)}}{\xi_{n_0(c)}}\right|
\mathbf 1_{A_T^c}\right]
\precsim
n_0(c)m^{-1/2}e^{-\kappa_1m/2}
+n_0(c)\left|\frac{\xi_m}{\xi_{n_0(c)}}-1\right|e^{-\kappa_1m}
=o(1).
\]
Thus, $\label{eq:mean_diff_max}
\mathbb E_0[(T_m-X_T)_+]=o(1).$

\noindent It remains to show \(\bigl|\mathbb E_0[X_T]-n_0(c)\bigr|=O(1)\). Using \(b_c\) as above,
\[
\bigl|\mathbb E_0[X_T]-n_0(c)\bigr|
\precsim
n_0(c)\left|
\frac{\mathbb E_0[V_T(\hat k_T)]-\xi_{n_0(c)}}{\xi_{n_0(c)}}
\right|
+n_0(c)|b_c-1|.
\]
The second term is \(O(1)\) by Assumption~\ref{ass:oracle_flatness_rate}. For the first,
\begin{align}
n_0(c)\left|
\frac{\mathbb E_0[V_T(\hat k_T)]-\xi_{n_0(c)}}{\xi_{n_0(c)}}
\right|
&\precsim
n_0(c)\left|
\mathbb E_0\!\left[
\frac{V_T(\hat k_T)-\xi_T}{\xi_T}\right]\right|
+n_0(c)\mathbb E_0\!\left|
\frac{\xi_T}{\xi_{n_0(c)}}-1\right|,
\label{eq:mean_diff_4}
\end{align}
where Assumption~\ref{ass:B1A} is used.

For the first term, let \(B_T=\{|T/n_0(c)-1|<\alpha\}\). By Lemma~\ref{lem:sign_bias_v} and \(T\ge m\),
\[
n_0(c)\left|
\mathbb E_0\!\left[\frac{V_T(\hat k_T)-\xi_T}{\xi_T}\right]\right|
\precsim
n_0(c)\mathbb E_0\!\left[\frac1T\right]
\le
\frac{1}{1-\alpha}
+\frac{n_0(c)}{m}e^{-\kappa_1m}
=O(1).
\]
For the second, choose \(0<\epsilon<(\alpha+\rho-1)/\rho\). Assumption~\ref{ass:B1A} and Lemma~\ref{lem:T_cont} give
\[
n_0(c)\mathbb E_0\!\left|
\frac{\xi_T}{\xi_{n_0(c)}}-1\right|
\le
n_0(c)\sup_{\left|t/\{\rho n_0(c)\}-1\right|\le\epsilon}
\left|\frac{\xi_t}{\xi_{n_0(c)}}-1\right|
+
\left(\sup_{t\ge1}\frac{\xi_t}{\xi_{n_0(c)}}+1\right)
n_0(c)e^{-\kappa_1m}
=O(1).
\]
Hence \(\left|\mathbb E_0[X_T]-n_0(c)\right|=O(1)\). Combining this with \eqref{eq:mean_diff_1} and \eqref{eq:mean_diff_max} gives $\left|\mathbb E_0[N_T]-n_0(c)\right|=O(1)$
which completes the proof.
\subsection{Proof of Theorem~\ref{th:risk_diff}}

\noindent By definition,
\begin{align}
R_{N_T}(\hat k_{N_T})-R_{n_0(c)}
&=
A\mathbb E_0\left[\frac{V_{N_T}(\hat k_{N_T})-\xi_{N_T}}{N_T}\right]
+A\mathbb E_0\left[\frac{\xi_{N_T}}{N_T}-\frac{\xi_{n_0(c)}}{n_0(c)}\right] \notag\\
&\quad+c\{\mathbb E_0[N_T]-n_0(c)\}.
\label{eq:risk_diff_decomp}
\end{align}
We control the three terms separately. First, conditioning on \(N_T\), Lemma~\ref{lem:sign_bias_v}, and Assumptions~\ref{ass:non_incr_tail_mass} and~\ref{as:A2} give
\begin{align}
\left|
A\mathbb E_0\left[\frac{V_{N_T}(\hat k_{N_T})-\xi_{N_T}}{N_T}\right]
\right|
&\precsim
A\mathbb E_0\left[\frac{\xi_{N_T}}{N_T^2}\right]
\le A\frac{\xi_m}{m^2}=o(1).
\label{eq:risk_diff_first}
\end{align}

\noindent For the second term, write
\begin{align}
\mathbb E_0\left[\frac{\xi_{N_T}}{N_T}-\frac{\xi_{n_0(c)}}{n_0(c)}\right]
&=
\mathbb E_0\left[\frac{\xi_{N_T}-\xi_{n_0(c)}}{N_T}\right]
+\xi_{n_0(c)}
\mathbb E_0\left[\frac1{N_T}-\frac1{n_0(c)}\right].
\label{eq:risk_diff_second_decomp}
\end{align}
Let \(\mathcal E_c=\{|N_T/n_0(c)-1|<\alpha\}\). On \(\mathcal E_c\), Assumption~\ref{ass:B1A} yields
\[
A\mathbb E_0\left[\left|
\frac{\xi_{N_T}-\xi_{n_0(c)}}{N_T}\right|\mathbf 1_{\mathcal E_c}\right]
\precsim \frac{A\xi_{n_0(c)}}{n_0(c)^2}\precsim c.
\]
On \(\mathcal E_c^c\), using \(N_T\ge m\), Assumption~\ref{ass:B1A}, Theorem~\ref{th:th_almost_sure_cons}, and the rate version of Lemma~\ref{lem:asymp_equi},
\[
A\mathbb E_0\left[\left|
\frac{\xi_{N_T}-\xi_{n_0(c)}}{N_T}\right|\mathbf 1_{\mathcal E_c^c}\right]
\precsim
\frac{A\xi_{n_0(c)}}{m}e^{-K_1m}
\precsim c\,\frac{n_0(c)^2}{m}e^{-K_1m}
=o(c).
\]
Hence
\begin{equation}
\label{eq:risk_diff_xi_part}
A\left|\mathbb E_0\left[
\frac{\xi_{N_T}-\xi_{n_0(c)}}{N_T}\right]\right|
\precsim c.
\end{equation}

\noindent It remains to control the reciprocal term. Define
\[
I_{3c}:=
\left|\mathbb E_0\left[\frac{n_0(c)}{N_T}-1\right]\right|.
\]
Using the same event \(\mathcal E_c\),
\begin{align}
I_{3c}
&\le
\frac{1}{n_0(c)(1-\alpha)}
\left|\mathbb E_0[N_T-n_0(c)]\right|
+\left(\frac{n_0(c)}{m}+1\right)\mathbb P_0(\mathcal E_c^c).
\label{eq:risk_diff_I3_bound}
\end{align}
By Theorem~\ref{th:mean_diff} and Theorem~\ref{th:th_almost_sure_cons}, the two terms are respectively \(O\{n_0(c)^{-1}\}\) and \(o\{n_0(c)^{-1}\}\). Thus \(I_{3c}=O\{n_0(c)^{-1}\}\). Using again the rate version of Lemma~\ref{lem:asymp_equi},
\[
A\xi_{n_0(c)}
\left|\mathbb E_0\left[\frac1{N_T}-\frac1{n_0(c)}\right]\right|
=O(c).
\]
Finally, Theorem~\ref{th:mean_diff} gives \(c|\mathbb E_0[N_T]-n_0(c)|=O(c)\). Combining these bounds in \eqref{eq:risk_diff_decomp},
\[
\left|R_{N_T}(\hat k_{N_T})-R_{n_0(c)}\right|=O(c),
\]
which completes the proof.
\section{Proof of Supporting Lemmas}\label{app_B}
\label{app:supporting_lemmas}

\begin{lemma} \label{lemm:trace_bound}\label{lem:trace_moment}
Under Assumption~\ref{as:spec_subg}, there exists a constant \(K_0>0\), depending only on \(\sigma_0\), such that for any \(t>0\),
{
\[
\mathbb{P}_{0}\!\left(\bigl|tr(S_n)-tr(\Sigma_0)\bigr|>t\right)
\leq
2\exp\!\left\{-nK_0\min\!\left(
\frac{t^2}{tr(\Sigma_0^2)},
\frac{t}{tr(\Sigma_0)}
\right)\right\}.
\]}
Moreover, for any integer \(k\ge1\), there exists \(C_k>0\), depending only on \(k,\sigma_0\), and \(\kappa_\sigma\), such that
{
\[
\mathbb E_0\!\left[\bigl|tr(S_n)-tr(\Sigma_0)\bigr|^k\right]
\le
C_k\left\{
\left(\frac{tr(\Sigma_0^2)}{n}\right)^{k/2}
+
\left(\frac{tr(\Sigma_0)}{n}\right)^k
\right\}.
\]}
\end{lemma}

\begin{proof}
The concentration bound follows directly from the Hanson--Wright inequality applied to the centered quadratic forms \(X_i^\top X_i-tr(\Sigma_0)\), followed by Bernstein's inequality for sums of independent sub-exponential random variables; see, e.g., \cite[Proposition~5.16]{vershynin2010introduction}. Here we use \(\|\Sigma_0\|_F^2=tr(\Sigma_0^2)\) and \(\|\Sigma_0\|\le tr(\Sigma_0)\).

For the moment bound, let \(Z_n=tr(S_n)-tr(\Sigma_0)\). Then
For the moment bound, let \(Z_n=tr(S_n)-tr(\Sigma_0)\). Then
{\small
\begin{align}
\mathbb E_0|Z_n|^k
&=
k\int_0^\infty t^{k-1}\mathbb P_0(|Z_n|>t)\,dt \notag\\
&\le
2k\int_0^\infty t^{k-1}
\exp\!\left\{-nK_0
\min\!\left(
\frac{t^2}{tr(\Sigma_0^2)},
\frac{t}{tr(\Sigma_0)}
\right)\right\}dt \notag\\
&\le
2k\int_0^\infty t^{k-1}
\exp\!\left\{-\frac{nK_0t^2}{tr(\Sigma_0^2)}\right\}dt
+
2k\int_0^\infty t^{k-1}
\exp\!\left\{-\frac{nK_0t}{tr(\Sigma_0)}\right\}dt \notag\\
&\le
C_k\left\{
\left(\frac{tr(\Sigma_0^2)}{n}\right)^{k/2}
+
\left(\frac{tr(\Sigma_0)}{n}\right)^k
\right\}.
\end{align}
}
where the last inequality follows from suitable changes of variables and the standard Gamma-integral identities.
\end{proof}

\begin{lemma}
\label{lem:moment_khat}
Let $\eta\in(0,1)$ be fixed and, for each sample size $n$, define $\hat{k}_n$ as the smallest integer such that the first $\hat{k}_n$ sample principal components explain at least a fraction $\eta$ of the total sample variance. Let $k_{0n}$ denote the corresponding population counterpart. Suppose Assumptions \ref{as:spec_subg} and \ref{ass:gap_k0} hold. Then, under \(\mathbb P_0\), \(\mathbb P_0(\hat{k}_n\neq k_{0n})\precsim \exp\{-nK_1\}\) for some constant \(K_1>0\). Moreover, for all sufficiently large \(n\) and any \(l\ge1\), \(\mathbb E_0[|\hat{k}_n-k_{0n}|^l]\precsim n^{-l}\).
\end{lemma}

\begin{proof}
Recall the population and sample cumulative proportions of explained variance, \(\gamma_k:=\sum_{j=1}^k\lambda_j(\Sigma_0)/tr(\Sigma_0)\) and \(\hat{\gamma}_k:=\sum_{j=1}^k\lambda_j(S_n)/tr(S_n)\), for \(k=0,1,\ldots,p_n\), with \(\gamma_0=\hat{\gamma}_0=0\). By Assumption \ref{ass:gap_k0}, it follows,
\begin{equation}\label{gap_prob_bound}
\mathbb{P}_0(\hat{k}_n\neq k_{0n})
\le
\mathbb{P}_0\!\left(\sup_{0\le k\le p_n}|\hat{\gamma}_k-\gamma_k|\ge g_0\right).
\end{equation}

Next, note that
\begin{align}\label{diff_gamma_bound}
|\hat{\gamma}_{k_{0n}}-\gamma_{k_{0n}}|
&=\Bigl|\frac{V_n(k_{0n})}{tr(S_n)}-\frac{\xi_n}{tr(\Sigma_0)}\Bigr|
\le
\frac{\sum_{j=1}^{k_{0n}}|\lambda_j(S_n)-\lambda_j(\Sigma_0)|
+|tr(\Sigma_0)-tr(S_n)|}{tr(\Sigma_0)}
\le 2\kappa_\sigma\|S_n-\Sigma_0\|.
\end{align}
Therefore,
\begin{equation}\label{gap_prob_bound_1}
\mathbb{P}_0(\hat{k}_n\neq k_{0n})
\le
\mathbb{P}_0\!\left(\|S_n-\Sigma_0\|\ge\frac{g_0}{2k_\sigma}\right)\precsim\exp(-K_1 n).
\end{equation}

\noindent The last step follows from concentration inequality for the sample covariance operator; see \cite[Theorem 1]{koltchinskii2017}. This proves the probability bound. For the moment bound, since $\operatorname{rank}(S_n)\le n$ and $\eta<1$, we have $\hat{k}_n\le n$ and by Assumption \ref{ass:gap_k0}, $k_{0n}=O(1)$.  Hence, for all sufficiently large $n$, $|\hat{k}_n-k_{0n}|\le n$, and therefore, we obtain
\[
\mathbb{E}\bigl[|\hat{k}_n-k_{0n}|^l\bigr] \le \mathbb{E}\bigl[||\hat{k}_n-k_{0n}|^l\mathbf{1}\{\hat{k}_n\neq k_{0n}\}\bigr] 
\le n^l\mathbb{P}_0(\hat{k}_n\neq k_{0n})
\le n^le^{-K_1n}\precsim n^{-l},
\]
for all sufficiently large $n$, which completes the proof.
\end{proof}

\begin{lemma}
\label{lem:rel_cont_v}
Let \(\hat{k}_n\) and \(k_{0n}\) be as defined in Lemma \ref{lem:moment_khat}, and suppose Assumptions~\ref{as:spec_subg}, \ref{ass:non_incr_tail_mass}, and~\ref{ass:gap_k0} hold. Then, for all sufficiently large \(n\), any \(\epsilon>0\), and any \(l\ge1\),
\begin{equation}
\small
\label{rel_cont_v5}
\mathbb{P}_0\!\left(\frac{|V_n(\hat{k}_n)-\xi_n|}{\xi_n}\ge\epsilon\right)
\le 2e^{-K_1n},
\qquad
\mathbb{E}_0\!\left[\left(\frac{|V_n(\hat{k}_n)-\xi_n|}{\xi_n}\right)^l\right]
\precsim n^{-l/2},
\end{equation}
for some constant \(K_1>0\).
\end{lemma}

\begin{proof}
Under Assumptions~\ref{ass:non_incr_tail_mass}, \(\xi_n\) is bounded away from zero and
\begin{align}
\label{rel_cont_v2}
\frac{|V_n(\hat{k}_n)-\xi_n|}{\xi_n}
\precsim
|tr(S_n)-tr(\Sigma_0)|+\|S_n-\Sigma_0\|
+|\hat{k}_n-k_{0n}|\{1+\|S_n-\Sigma_0\|\}.
\end{align}
Hence, for any \(\epsilon>0\) and some \(\epsilon^\ast>0\),
\begin{align}
\label{rel_cont_v4}
\mathbb{P}_0\!\left(\frac{|V_n(\hat{k}_n)-\xi_n|}{\xi_n}\ge\epsilon\right)
&\le
\mathbb{P}_0\!\left(|tr(S_n)-tr(\Sigma_0)|\ge\epsilon^\ast\right)
+\mathbb{P}_0\!\left(\|S_n-\Sigma_0\|\ge\epsilon^\ast\right)
+\mathbb{P}_0(\hat{k}_n\neq k_{0n})
\le 2e^{-K_1n},
\end{align}
for all sufficiently large \(n\), by Lemma~\ref{lem:trace_moment}, \cite[Theorem~1]{koltchinskii2017}, and Lemma~\ref{lem:moment_khat}.
for some constant $K_1>0$ depending only on $\kappa_\sigma$. The trace bound follows from Lemma~\ref{lem:trace_moment}, since Assumption~\ref{ass:gap_k0} gives \(k_{0n}=O(1)\), while Assumption~\ref{as:spec_subg} and the monotonicity of \(\{\xi_n\}\) imply \(tr(\Sigma_0)=\sum_{j=1}^{k_{0n}}\lambda_j(\Sigma_0)+\xi_n=O(1)\), and hence \(tr(\Sigma_0^2)=O(1)\). 
Next, by \eqref{rel_cont_v2}, for any \(l\ge1\),
\[
\mathbb{E}_0\!\left[\left(\frac{|V_n(\hat{k}_n)-\xi_n|}{\xi_n}\right)^l\right]
\precsim
\mathbb{E}_0|tr(S_n)-tr(\Sigma_0)|^l
+\mathbb{E}_0\|S_n-\Sigma_0\|^l
+\mathbb{E}_0|\hat{k}_n-k_{0n}|^l
+\Bigl\{\mathbb{E}_0|\hat{k}_n-k_{0n}|^{2l}\Bigr\}^{1/2}
\Bigl\{\mathbb{E}_0\|S_n-\Sigma_0\|^{2l}\Bigr\}^{1/2},
\]
where Cauchy--Schwarz is used for the last term. Lemma~\ref{lem:trace_moment}, \cite[Theorem~4]{koltchinskii2017}, and Lemma~\ref{lem:moment_khat} therefore yield \(\mathbb{E}_0[(|V_n(\hat{k}_n)-\xi_n|/\xi_n)^l]\precsim n^{-l/2}\) for all sufficiently large \(n\), completing the proof.
\end{proof}

\begin{lemma} \label{lem:sign_bias_v}
Let \(\hat{k}_n\) and \(k_{0n}\) be as defined in Lemma \ref{lem:moment_khat}, Under Assumptions~\ref{as:spec_subg}, \ref{ass:non_incr_tail_mass}, \ref{ass:gap_k0}, and~\ref{ass:spectral_gap_k0},
\[
\left|\mathbb E_0\left[\frac{V_n(\hat k_n)-\xi_n}{\xi_n}\right]\right|
\precsim n^{-1}.
\]
\end{lemma}
\begin{proof}
Let \(P_n=P_n^{k_0}\) and \(\widehat P_n^{k_0}\) denote the population and empirical rank-\(k_0\) PCA projections, respectively, and write \(E_n=S_n-\Sigma_0\). For the fixed-rank residual \(V_n(k_0)=tr\{(I-\widehat P_n^{k_0})S_n\}\), since \(\mathbb E_0E_n=0\), $\xi_n-\mathbb E_0V_n(k_0)
=\mathbb E_0\!\left[tr\{(\widehat P_n^{k_0}-P_n)S_n\}\right]\ge0,$
where the inequality follows from the Ky Fan variational characterization, since \(\widehat P_n^{k_0}\) maximizes \(tr(QS_n)\) over rank-\(k_0\) orthogonal projections \(Q\). To bound the bias, let \(Q\) be any such projection and write \(\Sigma_0=\sum_{j=1}^{p_n}\lambda_j(\Sigma_0)u_ju_j^\top\), \(P_n=\sum_{j=1}^{k_0}u_ju_j^\top\), and \(q_j=u_j^\top Q u_j\). Since \(\sum_jq_j=k_0\) and \(tr(P_nQ)=\sum_{j\le k_0}q_j\), $tr(P_n\Sigma_0)-tr(Q\Sigma_0) \ge \{\lambda_{k_0}(\Sigma_0)-\lambda_{k_0+1}(\Sigma_0)\}\{k_0-tr(P_nQ)\}.$
Moreover, \(\|Q-P_n\|_F^2=2\{k_0-tr(P_nQ)\}\); hence Assumption~\ref{ass:spectral_gap_k0} gives $tr(P_n\Sigma_0)-tr(Q\Sigma_0)\ge \frac{\gamma}{2}\|Q-P_n\|_F^2.$ Since $S_n=\Sigma_0+E_n$
\begin{equation}
\label{eq:empirical_projection_compare}
tr(QS_n)-tr(P_nS_n)
\le \|Q-P_n\|_F\|E_n\|_F-\frac{\gamma}{2}\|Q-P_n\|_F^2
\le \frac{\|E_n\|_F^2}{2\gamma},
\end{equation}
where the last inequality follows by maximizing \(ux-(\gamma/2)u^2\) over \(u\ge0\). Taking \(Q=\widehat P_n^{k_0}\) in \eqref{eq:empirical_projection_compare} it follows,
\begin{equation}
\label{eq:fixed_rank_bias_bound}
0\le \xi_n-\mathbb E_0V_n(k_0)
\le \frac{1}{2\gamma}\mathbb E_0\|E_n\|_F^2 \precsim n^{-1}.
\end{equation}
The last step follows from Assumption~\ref{as:spec_subg},\ref{ass:gap_k0}, \citep[Theorem~4]{koltchinskii2017}.
Since \(\xi_n\) is bounded away from zero under Assumptions~\ref{ass:non_incr_tail_mass},
Since \(\xi_n\) is bounded away from zero and $O(1)$, \eqref{eq:fixed_rank_bias_bound} gives \(\left|\mathbb E_0[(V_n(k_0)-\xi_n)/\xi_n]\right|\precsim n^{-1}\). Hence,
\begin{equation}
\label{eq:random_rank_decomp}
\left|\mathbb E_0\left[\frac{V_n(\hat k_n)-\xi_n}{\xi_n}\right]\right|
\precsim n^{-1}
+\mathbb E_0|\hat k_n-k_{0n}|
+\mathbb E_0\!\left[|\hat k_n-k_{0n}|\|S_n-\Sigma_0\|\right].
\end{equation}
Lemma~\ref{lem:moment_khat} and the Cauchy--Schwarz argument used after \eqref{rel_cont_v5} show that the last two terms are \(O(n^{-1})\), proving the result.
\end{proof}

\begin{lemma}[Almost sure consistency of the intermediate sample size]
\label{lem:T_cont}
Let \(T_m\) be the intermediate sample size in Algorithm~\ref{algo_1}. Under Assumptions~\ref{as:spec_subg}, \ref{ass:non_incr_tail_mass}, \ref{ass:gap_k0}, \ref{as:A1}, and~\ref{as:A2}, \(T_m/n_0(c)\xrightarrow{\mathrm{a.s.}}\rho\) under \(\mathbb P_0\) as \(c\downarrow0\).
\end{lemma}

\begin{proof}
Let \(\epsilon\in(0,1)\). From the definition of $T_m$ it follows,
{\small
\[
\mathbb P_0\!\left(\left|\frac{T_m}{\rho n_0(c)}-1\right|>\epsilon\right)
\le
\mathbb P_0\!\left(\sqrt{\frac{A}{c}V_m(\hat k_m)}<(1-\epsilon)n_0(c)\right)
+\mathbf 1\!\left\{\frac{m}{\rho n_0(c)}>1+\epsilon\right\}
+\mathbb P_0\!\left(\sqrt{\frac{A}{c}V_m(\hat k_m)}>(1+\epsilon)n_0(c)-\frac1\rho\right).
\]
}
By Assumption~\ref{as:A1}, \(m=o\{n_0(c)\}\), so the indicator vanishes for sufficiently small \(c\). Moreover, by the definition of \(n_0(c)\) and Assumption~\ref{as:A2}, $(n_0(c)^2c)/A
=\xi_{n_0(c)}\{1+o(1)\}
=\xi_m\{1+o(1)\},$ while \(\{(1+\epsilon)n_0(c)-1/\rho\}^2=(1+\epsilon)^2n_0(c)^2\{1+o(1)\}\). Hence, after possibly replacing \(\epsilon\) by a smaller positive constant,
\[
\mathbb P_0\!\left(\left|\frac{T_m}{\rho n_0(c)}-1\right|>\epsilon\right)
\le
\mathbb P_0\!\left(
\frac{|V_m(\hat k_m)-\xi_m|}{\xi_m}\ge\epsilon
\right)
\]
for all sufficiently small \(c\). The result is now an immediate consequence of Lemma~\ref{lem:rel_cont_v} and the Borel--Cantelli lemma.
\end{proof}

\bibliographystyle{cas-model2-names}
\bibliography{name.bib}

\end{document}